\documentclass[11pt,reqno]{amsart}
\usepackage{amsmath,amsthm}
\usepackage{latexsym}
\usepackage{amsfonts}
\usepackage{amssymb}
\usepackage[dvipsnames]{xcolor}
\usepackage{graphicx}
\usepackage{float}
\usepackage{graphbox}
\usepackage{textcomp}
\usepackage[pagebackref, colorlinks = true, linkcolor = blue, urlcolor  = blue, citecolor = red]{hyperref}
\usepackage[margin=1in]{geometry}
\usepackage{array}
\usepackage{comment}

\renewcommand{\epsilon}{\varepsilon}
\renewcommand{\phi}{\varphi}
\newcommand{\sa}{\textrm{sa}}

\newtheorem{theorem}{Theorem}[section]
\newtheorem{proposition}[theorem]{Proposition}
\newtheorem{lemma}[theorem]{Lemma}
\newtheorem{corollary}[theorem]{Corollary}
\newtheorem{definition}[theorem]{Definition}

\newtheorem{remark}[theorem]{Remark}

\newcommand{\ip}[2]{\left\langle\,#1\,|\,#2\,\right\rangle} 
\newcommand{\ket}[1]{|#1\rangle} 
\newcommand{\bra}[1]{\langle#1|} 
\newcommand{\kb}[2]{|#1\rangle\langle#2|} 
\newcommand{\bk}[2]{\langle#1|#2\rangle} 
\DeclareMathOperator{\Tr}{Tr}

\definecolor{candypink}{rgb}{0.89, 0.44, 0.48}

\begin{document}

\title[Multipartite entanglement detection via projective tensor norms]{Multipartite entanglement detection\\via projective tensor norms}

\author{Maria Anastasia Jivulescu}
\address{Department of Mathematics, Politehnica University of Timi\c soara, Victoriei Square 2, 300006 Timi\c soara, Romania.}
\email{maria.jivulescu@upt.ro}

\author{C\'ecilia Lancien}
\address{Institut de Math\'ematiques de Toulouse \& CNRS, Universit\'e Paul Sabatier, 118 route de Narbonne, F-31062 Toulouse Cedex 9, France.}
\email{clancien@math.univ-toulouse.fr}

\author{Ion Nechita}
\address{Laboratoire de Physique Th\'eorique \& CNRS, Universit\'e Paul Sabatier, 118 route de Narbonne, F-31062 Toulouse Cedex 9, France.}
\email{nechita@irsamc.ups-tlse.fr}

\begin{abstract}
We introduce and study a class of entanglement criteria based on the idea of applying local contractions to an input multipartite state, and then computing the projective tensor norm of the output. More precisely, we apply to a mixed quantum state a tensor product of contractions from the Schatten class $S_1$ to the Euclidean space $\ell_2$, which we call entanglement testers. We analyze the performance of this type of criteria on bipartite and multipartite systems, for general pure and mixed quantum states, as well as on some important classes of symmetric quantum states. We also show that previously studied entanglement criteria, such as the realignment and the SIC POVM criteria, can be viewed inside this framework. This allows us to answer in the positive two conjectures of Shang, Asadian, Zhu, and G\"uhne by deriving systematic relations between the performance of these two criteria.
\end{abstract}

\date{\today}

\maketitle

\tableofcontents

\section{Introduction}\label{sec:intro}

Determining whether a multipartite quantum system is in a \emph{separable} or an \emph{entangled} state is of prime importance in quantum information theory. Indeed, if such a quantum system is in a separable state, it means that there are no intrinsically quantum correlations between its subsystems, so that it is not providing any advantage compared to a classical system in information processing tasks. However, the problem of deciding if a multipartite quantum state is separable or entangled (and even only approximate versions of it) is known to be computationally hard \cite{Gharibian}. Standard solutions to overcome this practical difficulty consist in looking for necessary conditions to separability that would be easier to check than separability itself. These are usually dubbed \emph{entanglement criteria}, and various families of such criteria have already been extensively studied in the past.

From a mathematical point of view, a quantum state is described by a positive semidefinite operator on a complex Hilbert space having unit Schatten $1$-norm. And as we will see in more details later, for a quantum state on a multipartite system (i.e.~on a tensor product Hilbert space), being entangled is characterized by having a so-called \emph{projective Schatten $1$-norm} which is strictly larger than $1$ \cite{rudolph2000separability}. But there is no efficient way of estimating such tensor norm in general \cite{perez2004deciding}. An alternative consists in looking at other tensor norms, whose values are easier to compute and always smaller than the tensor norm characterizing entanglement (so that if they are strictly larger than $1$, then the state is guaranteed to be entangled). 

This is the approach that we take in this work. We define and study a class of entanglement criteria based on the idea of applying local contractions to an input multipartite state, and then computing the projective tensor norm of the output. More precisely, the local contractions that we consider are from the Schatten $1$-norm to the $\ell_2$-norm, i.e.~from a non-commutative space to a commutative one. This is what makes such entanglement criteria interesting in practice: they can be seen as reducing the study of mixed state entanglement to that of pure state entanglement, which is an easier task. 

Another advantage of our entanglement criteria is that their definition is independent from the number of subsystems. Several aspects are, admittedly, simpler to understand in the bipartite case, but they remain equally well-suited to the case where more than two parties are involved. In fact, one of the main issues with most well-known entanglement criteria is that they are specifically designed for bipartite systems, and generalizations to systems with more parties are not fully satisfying. Indeed, they usually consist in applying the bipartite criterion across all bipartitions, which certifies entanglement across bipartitions of the subsystems, but not genuinely multipartite entanglement \cite{horodecki2006separability}.

Well studied entanglement criteria, such as the realignment criterion \cite{chen2003realignment,rudolph2004ccnr} and the SIC POVM criterion \cite{shang2018enhanced}, are important examples in the framework we consider. Our work provides natural generalizations of these criteria to the multipartite setting, going beyond the biseparable case already discussed in the literature. Moreover, we establish an exact relation between the performance of the realignment and the SIC POVM criteria, solving in the positive two conjectures from \cite{shang2018enhanced}.

The remainder of the paper is organised as follows.  In Section \ref{sec:banach} we first recall basic facts about tensor norms on Banach spaces, and then relate them with the characterization of entanglement. With these observations at hand, we can define in Section \ref{sec:testers} the main objects of interest in this work, which we dub \emph{entanglement testers}, and establish some of their first key properties. Section \ref{sec:important-examples} is dedicated to providing explicit examples of testers, by showing that several well-known entanglement criteria (such as the celebrated realignment criterion or one based on SIC POVMs introduced more recently) can actually be seen as corresponding to a tester. In Section \ref{sec:perfect-testers} we define and characterize an important sub-class of testers: that of \emph{perfect testers}, which detect all entangled pure states. This naturally brings us to Section \ref{sec:symmetric-testers}, where we show that the examples of Section \ref{sec:important-examples} are in fact all special instances inside an important sub-class of perfect testers: that of \emph{symmetric testers}. From then on we focus for a while on these symmetric testers. We quantify how they perform in detecting the entanglement of several classes of bipartite states: pure states (Section \ref{sec:pure-states}), Werner and isotropic states (Section \ref{sec:symmetric-states}), pure states with white noise (Section \ref{sec:noisy-states}). On all these examples, we observe a systematic relation between the performance of the realignment and SIC POVM testers, proving on the way a conjecture from \cite{shang2018enhanced}. So in Section \ref{sec:conjecture} we ask whether it would hold more generally, for any bipartite state. This allows us to answer in the positive another conjecture from \cite{shang2018enhanced}. After this time spent on studying the specific case of symmetric testers we go back to a more general question in Section \ref{sec:completeness}, namely: is our family of entanglement criteria complete, i.e.~in other words, can any bipartite entangled state be detected by a tester? In Section \ref{sec:multipartite}, we take a look at what can be shown in the multipartite case. Finally, we present in Section \ref{sec:conclusions} an overview of the main results, as well as a list of the problems we have left open and some directions for future work. 

\section{Tensor products of Banach spaces and quantum entanglement}
\label{sec:banach}

We gather in this section basic definitions and facts about the different natural norms one can put on the algebraic tensor product of finite dimensional Banach spaces. Studying the different tensor norms of multipartite pure and mixed quantum states, and relating these norms to quantum entanglement, is the main theme of our work. In this sense, the current section contains the mathematical foundation on which the practical applications to quantum information are built upon. 

Let us start by recalling the definitions of the projective and the injective tensor norms for (finite dimensional) Banach spaces. 

\begin{definition}
Consider $m$ Banach spaces $A_1,\ldots,A_m$. For a tensor $x \in A_1 \otimes \cdots\otimes A_m$, we define its \emph{projective tensor norm}
\begin{equation}
\label{eq:def-projective-norm}\|x\|_\pi := \inf \left\{ \sum_{k=1}^r \|a_k^1\| \cdots \|a_k^m\| \, : \, r \in \mathbb N, \ a_k^i\in A_i,\ x = \sum_{k=1}^r a_k^1 \otimes \cdots \otimes a_k^m \right\},
\end{equation}
and its \emph{injective tensor norm}
\begin{equation}
\label{eq:def-injective-norm}\|x\|_\epsilon := \sup \left\{ \ip{\alpha^1 \otimes \cdots \otimes \alpha^m}{x} \, : \, \alpha^i \in A_i^*,\  \|\alpha^i\| \leq 1 \right\}.
 \end{equation}    
\end{definition}

It is immediate to see that the projective tensor norm can be equivalently defined as
\begin{equation} \label{eq:def-projective-norm-2}
\|x\|_\pi := \inf \left\{ \sum_{k=1}^r |\lambda_k| \, : \, r \in \mathbb N, \ x = \sum_{k=1}^r \lambda_k\, a_k^1 \otimes \cdots \otimes a_k^m, \ a_k^i\in A_i,\ \|a_k^i\|\leq 1 \right\}.
\end{equation}

The projective and injective norms are examples of \emph{tensor norms} (also known as \emph{reasonable cross-norms}): for simple tensors, we have 
$$\|a_1 \otimes \cdots \otimes a_m \|_\pi = \|a_1 \otimes \cdots \otimes a_m \|_\epsilon = \|a_1\| \cdots \|a_m\|,$$
and the same factorization property holds for the dual norms they induce on the tensor product of dual spaces $A_1^* \otimes \cdots A_m^*$. Moreover, the projective and the injective norms are dual to one another: for all $x \in A_1 \otimes \cdots \otimes A_m$, 
\begin{align*}
     \|x\|_{\pi} &= \sup_{\|\alpha\|_{A_1^* \otimes_\epsilon \cdots \otimes_\epsilon A_m^*} \leq 1} \ip{\alpha}{x} ,\\
    \|x\|_{\epsilon} &= \sup_{\|\alpha\|_{A_1^* \otimes_\pi \cdots \otimes_\pi A_m^*} \leq 1} \ip{\alpha}{x}.
\end{align*}
The last property of the projective and the injective tensor norms that we would like to mention is that they are extremal among tensor norms: for any other tensor norm $\|\cdot\|$ on $A_1 \otimes \cdots \otimes A_m$, we have 
$$\forall\ x \in A_1 \otimes \cdots \otimes A_m, \quad \|x\|_\epsilon \leq \|x\| \leq \|x\|_\pi.$$

The following fact will be crucial to the main definition from the next section. 
\begin{proposition}\label{prop:bound-norm-tensor-product-operators}
Consider $m$ linear operators $T_i : A_i \to B_i$ between Banach spaces $A_i, B_i$, $1\leq i\leq m$. Then, for any tensor norm on $B_1 \otimes \cdots \otimes B_m$,
$$\left\|\bigotimes_{i=1}^m T_i\right\|_{A_1 \otimes_\pi \cdots \otimes_\pi A_m \to B_1  \otimes \cdots \otimes B_m} = \prod_{i=1}^m \|T_i\|_{A_i\to B_i}.$$
\end{proposition}

\begin{proof}
For the sake of clarity, let us consider the case $m=2$, the proof in the general case being similar. Set $T := T_1 \otimes T_2 : A_1\otimes A_2 \to B_1 \otimes B_2$.
The extremal points of the unit ball of the projective tensor product $A_1\otimes_{\pi}A_2$ are products of the extremal points of the unit balls of the factors $A_1,A_2$. Now, on such a product input $a=a_1\otimes a_2$, the output $T(a)$ factors as $T_1(a_1)\otimes T_2(a_2)$. And the maximal inner product of such a product tensor with an element of $(B_1\otimes B_2)^*$ is attained on a product tensor. The result thus follows. 
\end{proof}

We shall now relate the different tensor norms discussed above to \emph{quantum entanglement}. First, let us recall that a \emph{multipartite pure quantum state} is a unit vector $\psi$ in a tensor product of complex Hilbert spaces $H_1 \otimes \cdots \otimes H_m$. Here, we shall always assume $m \geq 2$. Since we only consider finite dimensional Hilbert spaces, we make the identification $H_i\cong\mathbb C^{d_i}$ for $1\leq i\leq m$. The pure state $\psi$ is said to be \emph{separable} if it is a pure tensor:
$$\ket \psi = \ket{\psi_1} \otimes \cdots \otimes \ket{\psi_m}.$$

The Banach space structure we consider for each factor is $(H_i, \|\cdot\|_2)$, where each space comes equipped with its Euclidean norm. We write $\ell_2^d := (\mathbb C^d, \|\cdot\|_2)$. In the case of pure states, the relation between entanglement and tensor norms is obvious, and formally stated below.

\begin{proposition} \label{prop:sep-pure}
    A pure quantum state $\psi \in H_1 \otimes \cdots \otimes H_m$, $\|\psi\|_2=1$, is separable if and only if 
    $$\|\psi\|_\epsilon = \|\psi\|_\pi = 1.$$
\end{proposition}

Actually, the injective norm is closely related to a fundamental multipartite entanglement measure, the \emph{geometric measure of entanglement} \cite{shimony1995degree,wei2003geometric,zhu2010additivity}
\begin{equation} \label{eq:G}
G(\psi) := - \log \sup \left\{ |\ip{\phi_1 \otimes \cdots \otimes \phi_m}{\psi}|^2 \, : \, \phi_i\in H_i,\ \|\phi_i\|=1 \right\} = -2 \log \|\psi\|_\epsilon.
\end{equation}

Let us now move to the more general case of mixed quantum states, i.e.~of operators $\rho$ on $H_1\otimes\cdots\otimes H_m$ which are positive semidefinite and of unit trace. The Banach space structure we consider on each of the spaces $\mathcal B(H_i)\cong \mathcal M_{d_i}(\mathbb C)$ is that given by the \emph{Schatten $1$-norm} (or \emph{nuclear norm})
$$\|X\|_1 = \Tr\sqrt{X^*X}.$$
We write $S_1^{d} := (\mathcal M_d(\mathbb C), \|\cdot\|_1)$ for the complex Banach space. Since mixed quantum states are self-adjoint operators, we shall also consider the real Banach space
$$S_{1,sa}^d := (\mathcal M_d^{sa}(\mathbb C), \|\cdot\|_1).$$
Note that, in general, we have $S_{1,\sa}^{d} = S_1^{d} \cap \mathcal M_d^\sa(\mathbb C)$ (see e.g.~\cite[Section 1.3.2]{aubrun2017alice}). 

We recall the following fact, relating the separability problem for mixed quantum states to projective tensor norms. Although this is a well-known fact, we give the proof for the sake of completeness and in order to show-case the relation between the positivity properties of $\rho$ and its tensor norms. 

\begin{theorem}{\cite[Theorem 5]{rudolph2000separability} or \cite[Theorem 1.1]{perez2004deciding}}
	For a multipartite mixed quantum state $\rho \in \mathcal M_{d_1}(\mathbb C) \otimes \cdots \otimes \mathcal M_{d_m}(\mathbb C)$, $\rho \geq 0$, $\Tr \rho =1$, the following assertions are equivalent:
	\begin{enumerate}
	    \item $\rho$ is separable,
	    \item $\|\rho\|_{S_{1,sa}^{d_1} \otimes_\pi \cdots \otimes_\pi S_{1,sa}^{d_m}} = 1$,
	    \item $\|\rho\|_{S_1^{d_1} \otimes_\pi \cdots \otimes_\pi S_1^{d_m}} = 1$.
	\end{enumerate}
\end{theorem}
\begin{proof}
	Let us first show the implication $(1) \implies (2)$. Given a separable quantum state $\rho$, we have
	$$\rho = \sum_{k=1}^r p_k \ket{x_k^1}\bra{x_k^1} \otimes \cdots \otimes \ket{x_k^m}\bra{x_k^m},$$
	for a probability distribution $(p_k)_{k=1}^r$ and unit vectors $x_k^i \in \mathbb C^{d_i}$, $k \in [r]$, $i \in [m]$. Obviously, $\|\ket{x_k^i}\bra{x_k^i}\|_{S_{1,sa}^{d_i}} = \|x_k^i\|^2=1$, for every $k$ and $i$. So using the separable decomposition, we have
	$$\|\rho\|_{S_{1,sa}^{d_1} \otimes_\pi \cdots \otimes_\pi S_{1,sa}^{d_m}} \leq \sum_{k=1}^r p_k \prod_{i=1}^m \|\ket{x_k^i}\bra{x_k^i}\|_{S_{1,sa}^{d_i}} = 1.$$
	Recall that the projective tensor norm $S_{1,sa}^{d_1} \otimes_\pi \cdots \otimes_\pi S_{1,sa}^{d_m}$ is the largest cross norm on $\mathcal M^{sa}_{d_1}(\mathbb C) \otimes \cdots \otimes \mathcal M^{sa}_{d_m}(\mathbb C)$. So in particular it is larger than the norm $S_{1,sa}^{d_1\cdots d_m}$ on this space, i.e.
	$$\|\rho\|_{S_{1,sa}^{d_1} \otimes_\pi \cdots \otimes_\pi S_{1,sa}^{d_m}} \geq \|\rho\|_{S_{1,sa}^{d_1 \cdots d_m}} = \Tr \rho = 1.$$
	
	The implication $(2) \implies (3)$ is trivial:
	$$1=\|\rho\|_{S_{1,sa}^{d_1} \otimes_\pi \cdots \otimes_\pi S_{1,sa}^{d_m}} \geq \|\rho\|_{S_{1}^{d_1} \otimes_\pi \cdots \otimes_\pi S_{1}^{d_m}} \geq \|\rho\|_{S_{1}^{d_1\cdots d_m}} = \Tr \rho =1,$$
	where we have used the fact that the infimum in \eqref{eq:def-projective-norm} is taken over a smaller set of possible decompositions in the self-adjoint case. 
	
	Finally for the implication $(3) \implies (1)$, consider a decomposition 
	$$\rho = \sum_{k=1}^s a_k^1 \otimes \cdots \otimes a_k^m$$
	achieving the minimum in \eqref{eq:def-projective-norm}, and such that each term is non-zero. (Because of the equivalence between definitions \eqref{eq:def-projective-norm} and \eqref{eq:def-projective-norm-2}, the infimum is indeed attained in our case, as the sets $S_1^{d_1},\ldots,S_1^{d_m}$ are compact.) Above, we have $a^i_k \in \mathcal M_{d_i}(\mathbb C)$, not necessarily self-adjoint. We argue in the same way as before: 
	\begin{equation}\label{eq:entanglement-projective-norm-inequality}
	1 = \|\rho\|_{S_1^{d_1} \otimes_\pi \cdots \otimes_\pi S_1^{d_m}} = \sum_{k=1}^s \prod_{i=1}^k \|a_k^i\|_{S_1^{d_i}}  \geq  \sum_{k=1}^s \prod_{i=1}^k |\operatorname{Tr} a_k^i| \geq  \sum_{k=1}^s \prod_{i=1}^k \operatorname{Tr} a_k^i  = \operatorname{Tr} \rho =  1,
	\end{equation}
	where we have used the inequality $\|X\|_{S^1_d} \geq |\Tr X|$. 
	Hence, for each $k \in [s]$, $i \in [m]$, we have $\|a_k^i\|_{S_1^{d_i}} = |\operatorname{Tr} a_k^i|$, and thus $a_k^i = \omega_k^i b_k^i$ for some phase factor $\omega_k^i \in \mathbb C$, $|\omega_k^i|=1$, and positive semidefinite operators $b_k^i \in \mathcal M_{d_i}(\mathbb C)$. Let us define $\omega_k :=\prod_{i=1}^m \omega_k^i$ (satisfying $|\omega_k|=1$) and 
	$$\tau_k := \left| \prod_{i=1}^m \operatorname{Tr} a_k^i \right| = \prod_{i=1}^m \operatorname{Tr} b_k^i > 0.$$
	From \eqref{eq:entanglement-projective-norm-inequality}, we have 
	$$1 = \sum_{k=1}^s \tau_k = \left|\sum_{k=1}^s \omega_k \tau_k\right|,$$
	and thus the $\omega_k$'s are simultaneously equal to $\pm 1$. Since $\rho$ is positive semidefinite, they all must be equal to $+1$, and thus
	$$\rho = \sum_{k=1}^s b_k^1 \otimes \cdots \otimes b_k^m$$
	for positive semidefinite operators $b^i_k$, finishing the proof. 
\end{proof}

To summarize, we have seen that the entanglement of pure, resp.~mixed, quantum states is related to the projective tensor product of $\ell_2$, resp.~$S_1$, Banach spaces. This connection will be discussed at length in the next section.

\section{Entanglement testers}
\label{sec:testers}

We introduce in this section the main tool developed in this work, \emph{entanglement testers}. Mathematically, these are linear applications from the space of matrices (physically, mixed quantum states) to the space of vectors (physically, pure quantum states). We shall study these maps and their properties from the point of view of Banach spaces, so we shall endow the space of matrices with the Schatten $1$-norm $S_1$ and the set of vectors with the Euclidean norm $\ell_2$. In the context of quantum information theory, these are the natural norms for the vector spaces, when studying mixed and pure quantum states respectively.  

To a $n$-tuple of matrices $(E_1, \ldots, E_n) \in \left(\mathcal M_d(\mathbb C)\right)^n$, we associate the linear map
$$ \mathcal E: X\in \mathcal M_d(\mathbb C) \mapsto \sum_{k=1}^n \Tr(E_k^*X) \ket k \in \mathbb C^n, $$
where $\{\ket k\}_{k=1}^n$ is some fixed orthonormal basis of $\mathbb C^n$ (see Figure \ref{fig:E-tester}). In a similar manner, to a $n$-tuple of matrices $(F_1, \ldots, F_n) \in \left(\mathcal M_d(\mathbb C)\right)^n$ we associate the $\mathbb R$-linear map $\mathcal F: \mathcal M_d^{sa}(\mathbb C) \to \mathbb C^n$ defined in the obvious manner. Note that if the matrices $F_i$ are themselves self-adjoint, the map $\mathcal F$ is valued in $\mathbb R^n$. We introduce now the main definition of this paper. 

\begin{figure}[H]
    \centering
    \includegraphics{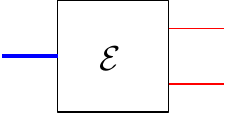}\qquad\qquad\qquad\includegraphics{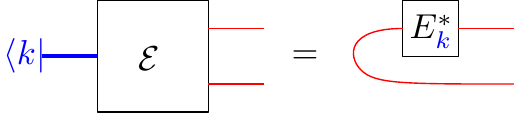}\qquad\qquad\qquad\includegraphics{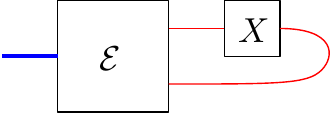}
    \caption{Left panel: an entanglement tester $\mathcal E$. The two wires (red, thin, dimension $d$) on the right are the input, while the wire on the left (blue, thick, dimension $n$) is the output. Center panel: the operator $\mathcal E$ in terms of the matrices $E_k$. Right panel: the vector $\mathcal E(X) \in \mathbb C^n$, for an input matrix $X \in \mathcal M_d(\mathbb C)$.}
    \label{fig:E-tester}
\end{figure}

\begin{definition} \label{def:tester}
	A $\mathbb C$-linear map $\mathcal E$ as above is called a $\mathbb C$-tester if $\|\mathcal E\|_{S_1^{d} \to \ell_2^{n}} = 1$. Similarly, an $\mathbb R$-linear map $\mathcal F$ is called an $\mathbb R$-tester if $\|\mathcal F\|_{S_1^{d,\sa} \to \ell_{2}^{n}} = 1$. 
\end{definition}

In the definition above, we distinguish between the real (self-adjoint) and the complex (general) cases. The following lemma shows that one can extend an $\mathbb R$-tester to a $\mathbb C$-tester. 

\begin{lemma}
Given an $\mathbb R$-linear map $\mathcal F:S_{1,\sa}^d \to \ell_{2}^n$, one can define its complexification
$$ \mathcal F' : X+\mathrm{i}Y \in S_1^d \mapsto \mathcal F(X) + \mathrm{i} \mathcal F(Y) \in \ell_2^n. $$

	We have $\|\mathcal F' \|_{S_1^{d} \to \ell_2^{n}} = \|\mathcal F \|_{S_1^{d,\sa} \to \ell_{2 }^{n}}$.
	In particular, if $\mathcal F$ is an $\mathbb R$-tester, then $\mathcal F'$ is a $\mathbb C$-tester.
\end{lemma}
\begin{proof}
	The inequality $\|\mathcal F'\| \geq \|\mathcal F\|$ is clear. For the converse, consider an extreme point $\kb{a}{b}$ of $S_1^d$ for which 
	$\|\mathcal F'\| = \|\mathcal F'(\kb{a}{b})\|$. We have 
	\begin{align*}
	\|\mathcal F'\|^2 &= \|\mathcal F'(\kb{a}{b})\|^2 \\
	& = \left\|\mathcal F'\left(\frac{\kb{a}{b}+\kb{b}{a}}{2}\right)\right\|^2 + \left\|\mathcal F'\left(\frac{\kb{a}{b}-\kb{b}{a}}{2\mathrm{i}}\right)\right\|^2\\
	&\leq \frac{\|\mathcal F\|^2}{4} \left( 2\|\kb{a}{b}\|_1^2+2\|\kb{b}{a}\|_1^2 \right)\\
	&= \|\mathcal F\|^2.
	\end{align*}
\end{proof}

In this paper, we shall focus mainly on the theory of $\mathbb C$-testers, to which we refer simply as (entanglement) \emph{testers}. In some sections (e.g.~Section \ref{sec:perfect-testers}) we shall want to differentiate between the self-adjoint case and the general one. To do this, we shall explicitly use the more precise notions of $\mathbb R$-testers and $\mathbb C$-testers.

We now look at tensor products of testers. Given $m$ sets of operators $E_i=\{E_{i;k}\}_{k=1}^{n_i}$, $1\leq i\leq m$, consider the respective maps
$$\mathcal E_i : X\in\mathcal M_{d_i}(\mathbb C) \mapsto \sum_{k=1}^{n_i} \mathrm{Tr}\left( E_{i;k}^* X \right) |k\rangle \in \mathbb C^{n_i}.$$
The tensor product of these $m$ maps acts on multipartite matrices $X \in \mathcal M_{d_1}(\mathbb C) \otimes \cdots \otimes \mathcal M_{d_m}(\mathbb C)$ as 
$$\mathcal E_1 \otimes \cdots \otimes \mathcal E_m(X)=\sum_{k_1=1}^{n_1}\ldots \sum_{k_m=1}^{n_m} \mathrm{Tr}\left( E_{1;k_1}^*\otimes \cdots \otimes E_{m;k_m}^* X \right) |k_1\cdots k_m\rangle.$$

Note that, from a physical perspective, the application $\mathcal E_1 \otimes \cdots \otimes \mathcal E_m$ maps mixed quantum states to pure quantum states, of possibly different dimensions (see Figure \ref{fig:testers}). This brings us to the the main theoretical insight of this section, the following corollary of Proposition \ref{prop:bound-norm-tensor-product-operators}.

\begin{figure}[H]
    \centering
    \includegraphics{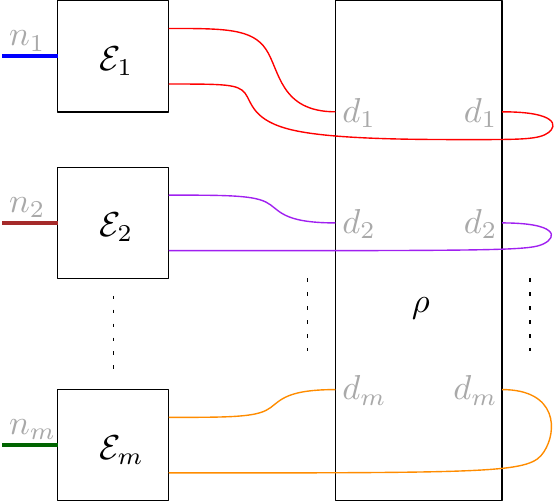}
    \caption{Applying a tensor product of entanglement testers $\mathcal E_1 \otimes \mathcal E_2 \otimes \cdots \otimes \mathcal E_m$ to a multipartite mixed quantum state $\rho$ results in a multipartite pure state.}
    \label{fig:testers}
\end{figure}

\begin{corollary} \label{cor:criterion}
Let $E_i=\{E_{i;k}\}_{k=1}^{n_i}$, $1\leq i\leq m$, be $m$ sets of operators as above, and let $\mathcal E_1, \ldots, \mathcal E_m$ be the corresponding linear maps. Then, for any $X \in \mathcal M_{d_1}(\mathbb C)\otimes\cdots\otimes \mathcal M_{d_m}(\mathbb C)$, we have
\[ \left\| \mathcal E_1 \otimes \cdots \otimes \mathcal E_m(X) \right\|_{\ell_2^{n_1} \otimes_\pi\cdots \otimes_\pi \ell_2^{n_m}} \leq \|\mathcal E_1\|_{S_1^{d_1} \to \ell_2^{n_1}} \cdots \|\mathcal E_m\|_{S_1^{d_m} \to \ell_2^{n_m}} \|X\|_{S_1^{d_1} \otimes_\pi \cdots \otimes_\pi S_1^{d_m}} \, . \]

In particular, if the $\mathcal E_i$'s are testers (real or complex), then for any multipartite quantum state $\rho$, the following implication holds:
\[ \rho \text{ separable} \ \Longrightarrow\ \left\| \mathcal E_1 \otimes \cdots \otimes \mathcal E_m(\rho) \right\|_{\ell_2^{n_1} \otimes_\pi\cdots \otimes_\pi \ell_2^{n_m}} \leq 1.\]

Reciprocally, we have the following \emph{entanglement criterion}: if the $\mathcal E_i$'s are testers, then
$$\left\| \mathcal E_1 \otimes \cdots \otimes \mathcal E_m(\rho) \right\|_{\ell_2^{n_1} \otimes_\pi\cdots \otimes_\pi \ell_2^{n_m}} >1 \implies \rho \text{ is entangled}.$$
\end{corollary}

In the rest of the paper, we shall study the \emph{power} of the entanglement criterion formulated above. We shall investigate which entangled states can be detected by a given family of testers, and which testers are best at detecting entanglement. Moreover, we shall see in the following sections that many know entanglement criteria fall into this framework. 

Let us now mention that, in the same way that the map $\mathcal E_1 \otimes \cdots \otimes \mathcal E_m$ gives an entanglement criterion, its inverse (assuming it exists) gives a \emph{separability criterion}. Indeed, using the same notation as above, and assuming that each map $\mathcal E_i$ is invertible, we have
$$ \|\rho\|_{S_1^{d_1} \otimes_\pi \cdots \otimes_\pi S_1^{d_m}} \leq \|\mathcal E_1^{-1}\|_{\ell_2^{n_1} \to S_1^{d_1}} \cdots \|\mathcal E_m^{-1}\|_{\ell_2^{n_m} \to S_1^{d_m}} \left\| [\mathcal E_1 \otimes \cdots \otimes E_m](\rho) \right\|_{\ell_2^{n_1} \otimes_\pi \cdots \otimes_\pi \ell_2^{n_m}}.$$
Hence, for any multipartite mixed quantum state $\rho$, the following implication holds 
\begin{equation}\label{eq:separability-criterion}
\left\| \mathcal E_1 \otimes \cdots \otimes \mathcal E_m(\rho) \right\|_{\ell_2^{n_1} \otimes_\pi\cdots \otimes_\pi \ell_2^{n_m}} \leq \frac{1}{\|\mathcal E_1^{-1}\|_{\ell_2^{n_1} \to S_1^{d_1}} \cdots \|\mathcal E_m^{-1}\|_{\ell_2^{n_m}\to S_1^{d_m}}} \ \Longrightarrow\  \rho \text{ is separable}.
\end{equation}
We postpone the discussion of these separability criteria to Section \ref{sec:important-examples}, where we shall see that they can only certify trivial separable states, so they are not useful in practice.

\subsection{Entanglement testers and their associated test operator}  \label{sec:test-operator} \hfill\smallskip

Given a set of operators $E=\{E_k\}_{k=1}^n$ on $\mathbb C^d$, let $\mathcal E:\mathcal{M}_d(\mathbb{C}) \to \mathbb{C}^n$ be the corresponding linear map, i.e.~
$$ \mathcal{E}:X\in \mathcal{M}_d(\mathbb{C}) \mapsto \sum_{k=1}^n \mathrm{Tr}(E_k^* X) \ket{k} \in \mathbb{C}^n . $$ 
We impose that $\mathcal E$ is a tester, as defined in Definition \ref{def:tester}. We recall that this means that $\|\mathcal{E}\|_{S_1^d\rightarrow\ell_2^n} = 1$, i.e.
\[ \max_{\|X\|_1\leq 1} \|\mathcal{E}(X)\|_2 =1 . \]
Now, given $X\in \mathcal{M}_d(\mathbb{C})$, we have
\[ \|\mathcal{E}(X)\|_2 = \left( \sum_{k=1}^n |\mathrm{Tr}(E_k^* X) |^2 \right)^{1/2} = \left( \sum_{k=1}^n \mathrm{Tr}( E_k^* \otimes E_k X\otimes X^* ) \right)^{1/2} . \]
Hence, setting 
\begin{equation} \label{eq:T_E} T_E:=\sum_{k=1}^n E_k\otimes E_k^* , \end{equation}
we want that
\begin{equation} \label{eq:norm-T_E} \max_{\|X\|_1\leq 1} \mathrm{Tr}( T_E^* X\otimes X^* ) = \max_{\|X\|_1\leq 1} \langle T_E,  X\otimes X^* \rangle =1 . \end{equation}
We call the operator $T_E$ on $\mathbb C^d\otimes\mathbb C^d$, defined in equation \eqref{eq:T_E}, the \emph{test operator} associated to the tester $\mathcal E$.

We would now like to characterize the set of test operators on $\mathbb C^d\otimes\mathbb C^d$. With this aim in view, given a tester $\mathcal E:S_1^d\rightarrow\ell_2^n$, let us define the completely positive map $\mathcal{T}_E:\mathcal{M}_d(\mathbb C) \to \mathcal{M}_d(\mathbb C)$ having the $E_k$'s as Kraus operators, i.e.
\[  \mathcal{T}_E:X\in\mathcal{M}_d(\mathbb C) \mapsto \sum_{k=1}^n E_kXE_k^* \in\mathcal{M}_d(\mathbb C) .\]
Then, denoting by $\Theta_E$ the Choi operator associated to $\mathcal{T}_E$, i.e.
\[ \Theta_E := \sum_{i,j=1}^d \mathcal{T}_E(\kb{i}{j}) \otimes \kb{i}{j}) ,\]
it is easy to check that we actually have
\begin{equation}\label{eq:Theta_E}
\Theta_E = \sum_{k=1}^n \kb{e_k}{e_k} , 
\end{equation}
where, for each $1\leq k\leq n$, $e_k\in\mathbb C^d\otimes\mathbb C^d$ is the vector version of $E_k\in\mathcal M_d(\mathbb C)$, i.e.~$\ket{e_k}=\sum_{i,j=1}^d \bra{i}E_k\ket{j}\ket{ij}$. Another way of writing this is, in terms of the operator $T_E$ defined in \eqref{eq:T_E}, is
\begin{equation} \label{eq:T-Theta} T_E=\Theta_E^{\Gamma}F , \end{equation}
where $\Gamma$ stands for the partial transposition and $F$ for the flip operator
\begin{align}
    \nonumber F : \mathbb C^d \otimes \mathbb C^d &\to \mathbb C^d \otimes \mathbb C^d\\
    \label{eq:def-flip} x \otimes y &\mapsto y \otimes x.
\end{align}

Yet another way of relating the operator $\Theta_E$ to the linear map $\mathcal E$ is via the relation 
$$\Theta_E = \mathcal E^*\mathcal E,$$
since we can re-write the application $\mathcal E$ as 
$$\mathcal E = \sum_{i=1}^k \kb{k}{e_k},$$
once we identify (as vector spaces) $\mathcal M_d(\mathbb C)$ with $\mathbb C^{d^2}$. Graphical representations of the operators $\Theta_E$ and $T_E$ are provided in Figure \ref{fig:Theta-T}.

\begin{figure}[H]
    \centering
    \includegraphics[align=c]{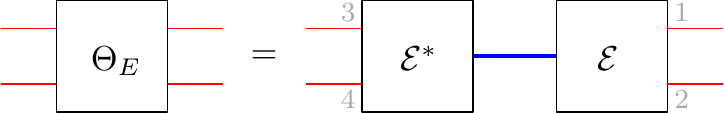}\qquad\qquad\includegraphics[align=c]{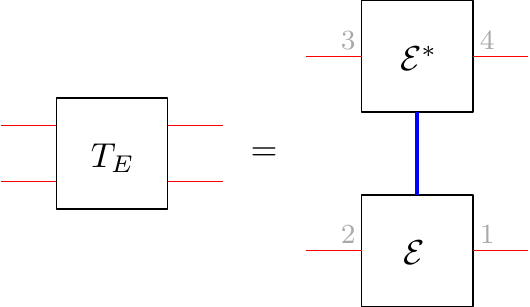}
    \caption{Graphical representations of the operators $\Theta_E, T_E \in \mathcal M_{d^2}(\mathbb C)$. Compare with equations \eqref{eq:Theta_E} and \eqref{eq:T_E}.}
    \label{fig:Theta-T}
\end{figure}

\begin{lemma} \label{lem:T-Theta}
The set of test operators on $\mathbb C^d\otimes\mathbb C^d$ is
\[ \left\{ \Theta^{\Gamma}F : \Theta\geq 0,\ \|\Theta\|_{S_{\infty,sa}^d\otimes_{\epsilon}S_{\infty,sa}^d} =1 \right\} . \]
\end{lemma}

\begin{proof}
We know from equation \eqref{eq:T-Theta} that $T$ is a non-normalized test operator on $\mathbb C^d\otimes\mathbb C^d$ if and only if $T=\Theta^{\Gamma}F$, where $\Theta$ is the Choi operator associated to a completely positive map on $\mathcal{M}_d(\mathbb C)$. By the Choi-Jamio\l{}kowski isomorphism, this is the same as saying that $\Theta$ is a positive semidefinite operator on $\mathbb C^d\otimes\mathbb C^d$. 

Let us now turn to the normalization condition given by \eqref{eq:norm-T_E}, i.e.
\[ \max\left\{ \mathrm{Tr}(T^* X\otimes X^*) : X\in S_{1}^d \right\} = 1 ,\]
By extremality in $S_1^d$ of rank one operators of the form $\kb{x}{y}$, where $x,y$ are unit vectors in $\mathbb C^d$, the latter is equivalent to
\[ \max\left\{ \bra{x\otimes y}T\ket{y\otimes x} : x,y\in\mathbb C^d,\ \|x\|_2=\|y\|_2= 1 \right\} = 1 . \]
In terms of $\Theta$, this reads
\[ \max\left\{ \bra{x\otimes \bar{y}}\Theta\ket{x\otimes \bar{y}} : x,y\in\mathbb C^d,\ \|x\|_2=\|y\|_2=1 \right\} = 1 . \]
By extremality in $S_{1,sa}^d$ of rank one operators of the form $\pm\kb{x}{x}$, where $x$ is a unit vector in $\mathbb C^d$, and because $\Theta$ is additionally positive semidefinite, the condition above is simply
\[ \max\left\{ \mathrm{Tr}(\Theta^* X\otimes Y) : X,Y\in S_{1,sa}^d \right\} = 1 ,\]
i.e.~by definition $\|\Theta\|_{S_{\infty,sa}^d\otimes_{\epsilon}S_{\infty,sa}^d} =1$.
\end{proof}

Note that Lemma \ref{lem:T-Theta} also provides a canonical way of constructing operators $\{E_k\}_{k=1}^{n}$ on $\mathbb C^d$ corresponding to a given test operator $T$ on $\mathbb C^d\otimes\mathbb C^d$. The strategy is to look at $\Theta=(TF)^{\Gamma}$, and diagonalize it as
\[ \Theta= \sum_{k=1}^n \lambda_k \kb{x_k}{x_k} , \]
with $1\leq n\leq d^2$, $\lambda_1,\ldots,\lambda_n>0$, $\{x_1,\ldots,x_n\}$ orthonormal family in $\mathbb C^d\otimes \mathbb C^d$. Then, we just have to define for each $1\leq k\leq n$, $E_k$ as being the matrix version of $\ket{e_k}:=\sqrt{\lambda_k}\ket{x_k}$. By construction, we have
\[ T=\sum_{k=1}^n E_k\otimes E_k^* . \]
This means that any test operator on $\mathbb C^d\otimes\mathbb C^d$ can be decomposed into a sum of at most $d^2$ terms of the form $E_k\otimes E_k^*$, where the $E_k$'s are orthogonal operators on $\mathbb C^d$.

\subsection{Equivalent testers} \hfill\smallskip

We consider now the notion of equivalent testers, characterizing pairs of testers which detect the same sets of entangled states. The definition below is motivated by the fact that the projective tensor norm on a tensor product of $\ell_2$ spaces is invariant by local unitary operators. So, applying such an operator to the output of a tensor product of testers does not change the outcome of the entanglement test. 

\begin{definition}\label{def:equivalent-testers}
    Two testers $\mathcal E, \mathcal F:S_1^d \to \ell_2^n$ are called \emph{equivalent} if there exists a unitary operator $U \in \mathcal U(n)$ such that, for all $X \in \mathcal M_d(\mathbb C)$, we have 
    $$\mathcal F(X) = U \mathcal E(X).$$
\end{definition}

A simple calculation shows that the operators $(F_j)_{j=1}^n$ defining the tester $\mathcal F$ are related to the operators $(E_k)_{k=1}^n$ defining $\mathcal E$ by the relation 
\begin{equation}\label{eq:Fj-equiv-Ek}
    \forall\ 1\leq j \leq n, \quad F_j = \sum_{k=1}^n \bar U_{jk} E_k.
\end{equation}

\begin{proposition}\label{prop:equal-T-equivalent}
    Two testers $\mathcal E, \mathcal F:S_1^d \to \ell_2^n$ are equivalent if and only if they have the same test operator.
\end{proposition}
\begin{proof}
    One direction is immediate: assuming that the operators $F_j$ are given by \eqref{eq:Fj-equiv-Ek}, we have 
    $$T_F = \sum_{j=1}^n F_j \otimes F_j^* = \sum_{j,k,l=1}^n \bar U_{jk} U_{jl} E_k \otimes E_l^* = \sum_{k=1}^n E_k \otimes E_k^* = T_E.$$
    
    For the other direction, note that $T_E = T_F$ implies $\Theta_E = \Theta_F$, hence the completely positive maps associated to the testers are identical: $\mathcal T_E = \mathcal T_F$. The conclusion follows from the fact that two different Kraus decompositions of a completely positive map are related by a unitary transformation as in \eqref{eq:Fj-equiv-Ek} (see \cite[Theorem 8.2]{nielsen2010quantum} or \cite[Corollary 2.23]{watrous2018theory}).
\end{proof}

\subsection{Practical interest} \label{sec:interest} \hfill\smallskip

Note that, in the bipartite case, an $\ell_2^n\otimes_{\pi}\ell_2^n$ norm can be simply seen as an $S_1^n$ norm. Indeed, let $\mathcal{E},\mathcal{F}:\mathcal{M}_{d}(\mathbb{C})\to \mathbb{C}^{n}$ be defined by 
\[ \mathcal{E}(X) = \sum_{k=1}^{n} \mathrm{Tr}(E_k^* X) \ket{k} \text{ and } \mathcal{F}(X) = \sum_{l=1}^{n} \mathrm{Tr}(F_l^* X) \ket{l} \, . \] 
We then have, for any $X \in \mathcal{M}_d(\mathbb{C}) \otimes \mathcal{M}_d(\mathbb{C})$,
\[ \mathcal{E}\otimes\mathcal{F}(X) = \sum_{k,l=1}^{n} \mathrm{Tr}( E_k^* \otimes F_l^*  X ) \ket{kl} \in \mathbb{C}^{n}\otimes\mathbb{C}^{n} \, , \]
which we can identify with
\[ \sum_{k,l=1}^n \mathrm{Tr}( E_k^*\otimes F_l^* X ) \kb{k}{l} \in \mathcal{M}_n(\mathbb{C}) \, . \]
In this way, the $\ell_2^n\otimes_{\pi}\ell_2^n$ norm of the first element is nothing else than the $S_1^n$ norm of the second one. Now, computing an $S_1^n$ norm is much cheaper than computing an $S_1^d \otimes_{\pi} S_1^d$ norm. The practical interest of our approach is thus clear in the bipartite case. 

But what about the difference in computational cost in the multipartite case? Assume that we have a multipartite system with $m$ subsystems. In this case deciding entanglement consists in computing an $(S_1^d)^{\otimes_{\pi}{m}}$ norm, i.e.~as we have just explained, an $(\ell_2^d\otimes_{\pi}\ell_2^d)^{\otimes_{\pi}{m}}$ norm. Now, an important property of the projective norm is that it is associative: given Banach spaces $X_1,X_2,X_3$, $X_1\otimes_{\pi}X_2\otimes_{\pi}X_3 = (X_1 \otimes_{\pi} X_2) \otimes_{\pi} X_3 = X_1 \otimes_{\pi} (X_2 \otimes_{\pi} X_3)$. Hence, an $(\ell_2^d\otimes_{\pi}\ell_2^d)^{\otimes_{\pi}m}$ norm can actually be seen as an $(\ell_2^d)^{\otimes_{\pi}2m}$ norm. On the other hand, deciding whether maps $\mathcal{E}_1,\ldots,\mathcal{E}_m$ detect entanglement here consists in computing an $(\ell_2^n)^{\otimes_{\pi} m}$ norm. This means that what we gain with our approach is a factor $2$ in the number of tensor products. 

In addition to being associative, the projective norm is also commutative: given Banach spaces $X_1,X_2$, $X_1\otimes_{\pi}X_2 = X_2 \otimes_{\pi} X_1$. This means that, for any state $\rho$ on $(\mathbb C^d)^{\otimes m}$, for any permutation $\sigma\in\mathcal S_{2m}$, denoting by $\rho_{\sigma}$ the matrix obtained from the matrix $\rho$ by permuting its indices according to $\sigma$, we have
\[ \|\rho\|_{(\ell_2^d)^{\otimes_{\pi}2m}} = \|\rho_{\sigma}\|_{(\ell_2^d)^{\otimes_{\pi}2m}} . \]
We could thus enlarge even further our family of entanglement criteria to: if a state $\rho$ on $(\mathbb C^d)^{\otimes m}$ is separable, then for any testers $\mathcal E_1,\ldots,\mathcal E_m:S_1^d\to\ell_2^n$ and any permutation $\sigma\in\mathcal S_{2m}$,
\[ \|\mathcal E_1\otimes\cdots\otimes \mathcal E_m(\rho_{\sigma}) \|_{(\ell_2^n)^{\otimes_{\pi}m}} \leq 1 . \]

Considering arbitrary permutations has however one important drawback: one loses the local aspect of the tester maps from Definition \ref{def:tester}. It is in this extended sense that we shall prove the completeness of entanglement criteria defined by testers for mixed bipartite states in Section \ref{sec:completeness}. 

It is important to mention the case of the permutations corresponding to partial transpositions. Indeed, if $I \subseteq [m]$ is the set of indices that are partially transposed and $\rho^{\Gamma_I}$ is the corresponding matrix, we have the following equality:
$$\|\mathcal E_1\otimes\cdots\otimes \mathcal E_m(\rho^{\Gamma_I})\|_{(\ell_2^n)^{\otimes_{\pi}m}} = \|\mathcal E_1'\otimes\cdots\otimes \mathcal E_m'(\rho)\|_{(\ell_2^n)^{\otimes_{\pi}m}},$$
where $\mathcal E_1',\ldots,\mathcal E_m'$ are defined as
$$\mathcal E_i':= \begin{cases}
\mathcal E_i^\flat \quad \text{if } i \in I\\
\mathcal E_i \quad \text{if } i \notin I
\end{cases} ,$$
with $\mathcal E_i^\flat$ the tester whose operators are the transposition of those of $\mathcal E_i$, so that $\mathcal E_i^\flat$ acts as $\mathcal E_i^\flat(X) = \mathcal E_i(X^\top)$ (see Figure \ref{fig:EF-X}).

\begin{figure}[H]
    \centering
    \includegraphics{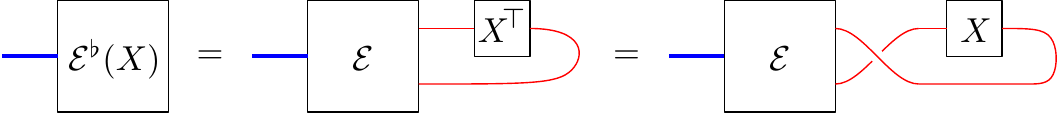}
    \caption{Implementing the transposition with a tester.}
    \label{fig:EF-X}
\end{figure}

\section{Important examples of testers}\label{sec:important-examples}

We now show that several well-known entanglement criteria can actually be seen as part of our framework, i.e.~as being associated to an entanglement tester $\mathcal E$.  
The way they are usually presented, all these examples appear as being designed for bipartite systems only. However, with our point of view, it is clear that they can naturally be extended to any number of subsystems (see Section \ref{sec:multipartite}).

\subsection{Maps defined from matrix bases} \hfill\smallskip

The first important example that enters into our framework is the celebrated realignment criterion \cite{chen2003realignment,rudolph2004ccnr}. Given an orthonormal basis $\{\ket{i}\}_{i=1}^d$ of $\mathbb{C}^d$, the set $R=\{R_{ij}\}_{i,j=1}^d = \{ |i\rangle \langle j|\}_{i,j=1}^d$ defines a map 
\[ \mathcal{R}: X \mapsto \sum_{i,j=1}^d \bra{i}X\ket{j} \ket{ij}. \]
The $\ell_2^{d^2}$ norm of $\mathcal{R}(X)$ is thus the $S_2^d$ norm of $X$, hence $\|\mathcal R(X)\|_{\ell_2^{d^2}} = \|X\|_{S_2^d} \leq \|X\|_{S_1^d}$, which implies that $\|\mathcal R\|_{S_1^d \to \ell_2^{d^2}} \leq 1$. Moreover, considering $X = \kb{x}{x}$ for a unit vector $x \in \mathbb C^d$ shows that actually $\|\mathcal R\|_{S_1^d \to \ell_2^{d^2}} = 1$. Thus $\mathcal R$ is a tester. And conversely, 
\[ \mathcal{R}^{-1}: x \mapsto \sum_{i,j=1}^d x_{ij} \kb{i}{j} \, , \]
so that $\|\mathcal R^{-1}\| = \sqrt d$. In other words, this means that, with the proper identification $\mathcal R = \mathrm{id}$, these norm estimates read $\|\mathrm{id}\|_{S_1^d \to S_2^d} = 1$ and $\|\mathrm{id}\|_{S_2^d \to S_1^d} = \sqrt d$ (see Figure \ref{fig:R-S}).

Note that the test operator $T_R$ associated to the realignment map $\mathcal{R}$ is the flip operator $F$ from equation \eqref{eq:def-flip}. Indeed,
$$T_R= \sum_{i,j=1}^d R_{ij}\otimes R_{ij}^* = \sum_{i,j=1}^d \kb{i}{j}\otimes\kb{j}{i} = F. $$

One can generalise the previous example to so-called cross-norm criteria in other local matrix bases than the one of matrix units. Given any orthonormal basis $G=\{ G_k \}_{k=1}^{d^2}$ of $\mathcal M_d(\mathbb C)$, the corresponding map is defined as
\[ \mathcal G : X \mapsto \sum_{k=1}^{d^2} \mathrm{Tr}(G_k^* X) \ket{k} \, . \]
Just as $\mathcal R$, it is such that the $\ell_2^{d^2}$ norm of $\mathcal{G}(X)$ is the $S_2^d$ norm of $X$, so that $\|\mathcal G\| = 1$. And conversely, 
\[ \mathcal{G}^{-1}: x \mapsto \sum_{k=1}^{d^2} x_{k} G_k\, , \]
so that $\|\mathcal G^{-1}\| = \sqrt d$. Hence exactly as for the map $\mathcal R$, the map $\mathcal G$ can be seen as the identity map from $S_1^d$ to $\ell_2^{d^2}\cong S_2^d$. The testers $\mathcal R$ and $\mathcal G$ are equivalent in the sense of Definition \ref{def:equivalent-testers}.

In \cite{sarbicki2020correlation}, deformed versions of these criteria, based on observed correlations in local matrix bases, were studied. This family of criteria actually enters in our framework as well. Let us briefly explain how. Let $G=\{ G_k \}_{k=1}^{d^2}$ be an orthonormal basis of $\mathcal M_d(\mathbb C)$, fix $x\geq 0$, and define
\[ \widetilde{\mathcal G}_x : X \mapsto x\mathrm{Tr}(G_1^* X) \ket{1} + \sum_{k=2}^{d^2} \mathrm{Tr}(G_k^* X) \ket{k} \, . \]
Assume now that $G=\{ G_k \}_{k=1}^{d^2}$ is a canonical orthonormal basis, i.e.~with one of its elements proportional to the identity, here $G_1=I/\sqrt{d}$, and the others traceless. We then have, for any $X$ such that $\|X\|_1\leq 1$,
\begin{align*} 
\|\widetilde{\mathcal G}_x(X)\|_2 & = \left( \sum_{k=1}^{d^2} |\mathrm{Tr}(G_k^* X)|^2 - (1-x^2)|\mathrm{Tr}(G_1^* X)|^2 \right)^{1/2} \\
& = \left( \mathrm{Tr}(XX^*) - \frac{1-x^2}{d}|\mathrm{Tr} (X)|^2 \right)^{1/2} \\
& \leq \left( 1- \frac{1-x^2}{d} \right)^{1/2} \, . 
\end{align*}
Hence, the map 
\[ \mathcal G_x:= \left( \frac{d}{d-1+x^2} \right)^{1/2} \widetilde{\mathcal G}_x \,  \]
is such that $\|\mathcal G_x\|=1$, and thus provides an entanglement criterion. This is in fact nothing else than a rephrasing of \cite[Theorem 1]{sarbicki2020correlation}.

\subsection{Maps defined from \texorpdfstring{$2$}{2}-designs} \hfill\smallskip

We discuss in this section testers coming from spherical 2-designs; an important special case corresponds to the entanglement criterion based on SIC POVMs introduced in \cite[Section IV]{shang2018enhanced} (see also \cite{lai2018entanglement}).
Given a spherical $2$-design $\{ |x_k \rangle\}_{k=1}^{d^2}$ of $\mathbb{C}^d$ with $d^2$ elements (see equation \eqref{defsic} below for the definition), the set $S = \left\{ S_k := \sigma |x_k \rangle \langle x_k| \right \}_{k=1}^{d^2}$ defines a map 
\[ \mathcal{S}: X \mapsto \sigma \sum_{k=1}^{d^2} \bra{x_k}X\ket{x_k} \ket{k}. \]
It is such that
\begin{align} 
\nonumber\|\mathcal{S}(X)\|_2 & = \sigma \left( \sum_{k=1}^{d^2} |\bra{x_k}X\ket{x_k}|^2 \right)^{1/2} \\
\nonumber& = \sigma \left( \sum_{k=1}^{d^2} \mathrm{Tr}\left( \kb{x_k}{x_k}^{\otimes 2}X\otimes X^*\right) \right)^{1/2} \\
\label{eq:norm-S-X}& = \sigma \left( \frac{2d}{d+1} \mathrm{Tr}\left(\frac{I+F}{2} X\otimes X^*\right) \right)^{1/2}, 
\end{align}
where the last equality is because, by definition of a spherical $2$-design, 
\begin{equation}\label{defsic}
 \frac{1}{d^2} \sum_{k=1}^{d^2} \kb{x_k}{x_k}^{\otimes 2} = \frac{I+F}{d(d+1)}. \end{equation} 
This implies that $\|\mathcal S\| = \sigma\sqrt{2d/(d+1)}$, so in order to obtain the correct normalization for the map $\mathcal S$, one needs to fix 
\[ \sigma=\sqrt{\frac{d+1}{2d}}. \]
Let us now compute $\|\mathcal{S}^{-1}\|$. We have $\mathcal S^{-1} = \mathcal S^* (\mathcal S \mathcal S^*)^{-1}$, with
\begin{align*}
& \mathcal S = \sum_{k=1}^{d^2} \mathrm{Tr}(S_k^* \cdot)\ket{k} ,\\
& \mathcal S^* = \sum_{k=1}^{d^2} S_k\bra{k}.
\end{align*}
Using the symmetry of the $S_k$'s, the Gram matrix $G=\mathcal S \mathcal S^*$ is easily computed as
$$G = \sum_{k,l=1}^{d^2} \mathrm{Tr} (S_k S_l) \ket{k} \bra{l} = \frac{1}{2d}J + \frac 1 2 I,$$
where $J$ is the all ones $d^2 \times d^2$ matrix. The inverse of $G$ is thus
$$G^{-1} = 2I - \frac{2d}{d+1} \ket{v}\bra{v},$$
where the unit vector $v\in \mathbb C^{d^2}$ is defined as
\begin{equation}\label{eq:def-v}
\ket v = \frac 1 d \sum_{k=1}^{d^2}\ket{k}.
\end{equation}
We thus have
$$\mathcal S^{-1} = 2 \mathcal S^* - \sqrt\frac{2d}{d+1} \ket I \bra v,$$
where $\ket I$ is the vectorization of the $d \times d$ identity matrix. So for a given $y \in \mathbb C^{d^2}$, we have 
$$\mathcal S^{-1}(\ket{y}) = 2 \sum_{k=1}^{d^2}y_kS_k - \sqrt\frac{2}{d(d+1)}\left(\sum_{k=1}^{d^2}y_k\right) I = \sum_{k=1}^{d^2} y_kM_k,$$
where we have defined the $d\times d$ matrices $M_k$ as
$$M_k:\,=  \sqrt{\frac{2}{d}}\left( \sqrt{d+1}\kb{x_k}{x_k} -\frac{1}{\sqrt{d+1}}I \right) .$$
Hence,
\begin{align*}
\underset{\|y\|_2\leq 1}{\max} \left\| \sum_{k=1}^{d^2}y_kM_k \right\|_1 
& = \underset{\|y\|_2\leq 1}{\max} \underset{\|Y\|_{\infty}\leq 1}{\max} \left| \sum_{k=1}^{d^2} y_k \mathrm{Tr}(M_kY) \right| \\
& =  \underset{\|Y\|_{\infty}\leq 1}{\max} \underset{\|y\|_2\leq 1}{\max} \left| \sum_{k=1}^{d^2} y_k \mathrm{Tr}(M_kY) \right| \\
& = \underset{\|Y\|_{\infty}\leq 1}{\max} \left( \sum_{k=1}^{d^2} |\mathrm{Tr}(M_kY)|^2 \right)^{1/2} \\
& = \underset{\|Y\|_{\infty}\leq 1}{\max} \left( \sum_{k=1}^{d^2} \mathrm{Tr}(M_k\otimes M_k^* Y\otimes Y^*) \right)^{1/2} .
\end{align*}
Now by definition, 
\[ M_k\otimes M_k^* = \frac{2}{d}\left( (d+1)\kb{x_k}{x_k}^{\otimes 2}- \kb{x_k}{x_k}\otimes I - I\otimes\kb{x_k}{x_k} + \frac{1}{d+1}I \right). \]
And therefore,
\[ \sum_{k=1}^{d^2} M_k\otimes M_k^* = \frac{2}{d}\left( d(I+F)-2dI+\frac{d^2}{d+1}I \right) = 2\left( F-\frac{1}{d+1}I \right). \]
This implies that
\begin{align*}
\underset{\|Y\|_{\infty}\leq 1}{\max} \left( \sum_{k=1}^{d^2} \mathrm{Tr}(M_k\otimes M_k^* Y\otimes Y^*) \right)^{1/2} 
& = \underset{\|Y\|_{\infty}\leq 1}{\max} \left( 2 \mathrm{Tr}\left( \left(F-\frac{1}{d+1}I \right) Y\otimes Y^*\right) \right)^{1/2} \\
& = \sqrt{2} \underset{\|Y\|_{\infty}\leq 1}{\max} \left(  \mathrm{Tr}(YY^*)-\frac{1}{d+1}|\mathrm{Tr}(Y)|^2\right)^{1/2} \\
& =\sqrt{2d}.
\end{align*}
And we have thus eventually shown that $\|\mathcal{S}^{-1}\|=\sqrt{2d}$.

Note that the test operator $T_S$ associated to the $2$-design map $\mathcal{S}$ is the projector on the symmetric subspace $(I+F)/2$. Indeed, by definition of a $2$-design, as recalled in equation \eqref{defsic}, we have
$$T_S= \sum_{k=1}^{d^2} S_{k}\otimes S_{k}^* = \frac{d+1}{2d} \sum_{k=1}^{d^2} \kb{x_k}{x_k}\otimes\kb{x_k}{x_k} = \frac{I+F}{2},$$
which is the projection on the symmetric subspace of $\mathbb C^d \otimes \mathbb C^d$ (see also Figure \ref{fig:R-S}).

\begin{figure}[H]
    \centering
    \includegraphics{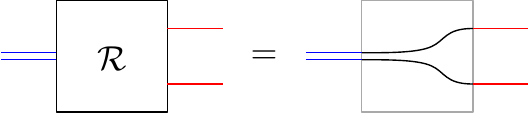}\qquad\qquad\qquad\qquad\includegraphics{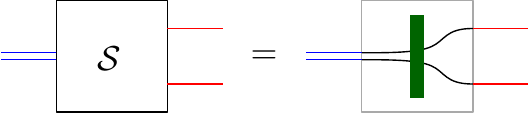}
    \caption{The diagrams for the realignment tester $\mathcal R$ and the SIC POVM tester $\mathcal S$. The (green) solid rectangle in the diagram for $\mathcal S$ is the symmetrization operation in Penrose's graphical formalism. The two blue output wires are treated as a single vector space $\mathbb C^n = \mathbb C^{d^2}$.}
    \label{fig:R-S}
\end{figure}
 
\subsection{Separability criteria}  \hfill\smallskip

Recall from Section \ref{sec:testers} that linear maps $\mathcal E$ also provide separability criteria (see equation \eqref{eq:separability-criterion}). It follows that maps $\mathcal E:S_1^d \to \ell_2^{n}$ with minimal $\| \mathcal E\| \cdot \| \mathcal E^{-1}\|$ should give the best combined (entanglement + separability) criterion. 

The two examples that we will be mostly focusing on in this paper are the matrix unit map $\mathcal R$ and the $2$-design map $\mathcal S$. Note from the above discussion that, for $\mathcal R,\mathcal S$ normalized to have norm $1$, we have $\|\mathcal{S}^{-1}\|=\sqrt{2d} > \sqrt{d} = \|\mathcal{R}^{-1}\|$. So the $2$-design map gives a `worse' combined (entanglement + separability) criterion than the matrix unit map, in the sense that
\[ \|\mathcal S\| \cdot \|\mathcal{S}^{-1}\| > \|\mathcal R\| \cdot \|\mathcal{R}^{-1}\| \, .  \] 
Note however that this goes in the opposite direction as the one suggested by the observations in \cite{shang2018enhanced}.

\begin{remark}
The Banach-Mazur distance between Banach spaces $X,Y$ is defined as the infimum over all maps $\mathcal{T}:X\to Y$ of $\|\mathcal{T}\|\cdot\|\mathcal{T}^{-1}\|$. In \cite[Theorem 45.2]{tomczak1989banach} it is shown that the Banach-Mazur distance between $S_1^d$ and $\ell_2^{d^2} \cong S_2^d$ is $\sqrt d$. This means that the realignment criterion is the `best' possible, when one is looking for simultaneous entanglement and separability criteria.
\end{remark}

The separability criterion defined by the realignment map has very little practical interest, as we shall see next for the case of bipartite quantum states. A quantum state $\rho \in \mathcal M_d(\mathbb C) \otimes \mathcal M_d(\mathbb C)$ is certified separable by the realignment criterion from \eqref{eq:separability-criterion} whenever 
$$\|\mathcal R^{\otimes 2}(\rho)\|_1  \leq \frac{1}{d}.$$
Note however that we have 
$$\|\mathcal R^{\otimes 2}(\rho)\|_{1} \geq \|\mathcal R^{\otimes 2}(\rho)\|_2 = \|\rho\|_2 \geq \frac{1}{d}\|\rho\|_1 = \frac 1 d,$$
with equality if and only if the spectrum of $\rho$ is flat, i.e.~$\rho = I/d^2$. Hence, the only bipartite state which is certified separable by the realignment criterion is the maximally mixed state. 

\section{Perfect testers}\label{sec:perfect-testers}

In this section, we introduce and study a special class of entanglement testers, called perfect testers, which are strong enough to detect any pure entangled state.

\subsection{Definition and characterization of perfect testers} \hfill\smallskip

\begin{definition}
	Let $d \geq 2$. A $\mathbb C$-tester $\mathcal E : S_1^{d} \to \ell_2^{n}$ is called $\mathbb C$-\emph{perfect} if, for any pure states $\phi,\chi \in \mathbb C^{d} \otimes \mathbb C^{d}$, at least one of them entangled, 
	$$\|\mathcal E ^{\otimes 2}(\ket{\phi}\bra{\chi})\|_{\ell_2^{n} \otimes_\pi \ell_2^{n}} >1.$$
	  
	An $\mathbb R$-tester $\mathcal F : S_1^{d,\sa} \to \ell_2^{n}$ is called $\mathbb R$-\emph{perfect} if, for any pure entangled state $\phi \in \mathbb C^{d} \otimes \mathbb C^{d}$, 
	$$\|\mathcal F ^{\otimes 2}(\ket{\phi}\bra{\phi})\|_{\ell_2^{n} \otimes_\pi \ell_2^{n}} >1.$$
\end{definition}

We state and prove below the main result of this section, a characterization of perfect testers, both in the real and the complex cases.

\begin{theorem} \label{th:perfect-testers}
	Consider a $\mathbb C$-linear map $\mathcal E : S_1^{d} \to \ell_2^{n}$. The following statements are equivalent:
	\begin{enumerate}
	    \item $\mathcal E$ is a $\mathbb C$-perfect tester,
	    \item The norm $\|\mathcal E \|_{S_1^{d} \to \ell_2^{d}} = 1$ is attained at all the extremal points of the unit ball of $S_1^{d}$: for all unit vectors $x,y\in \mathbb C^{d}$ we have $\|\mathcal E(\kb{x}{y})\|_2 = 1$,
	    \item $\mathcal E$ is an isometry $S_2^{d} \to \ell_2^{n}$.
	\end{enumerate}
	
	Similarly, for an $\mathbb R$-linear map $\mathcal F : S_1^{d,sa} \to \ell_2^{n}$, the following are equivalent: 
	\begin{enumerate}
	    \item $\mathcal F$ is an $\mathbb R$-perfect tester,
	    \item The norm $\|\mathcal F \|_{S_1^{d,sa} \to \ell_2^{n}} = 1$ is attained at all the extremal points of the unit ball of $S_1^{d,sa}$: for all unit vector $x\in \mathbb C^{d}$ we have $\|\mathcal E(\kb{x}{x})\|_2 = 1$.
	\end{enumerate}
\end{theorem}
\begin{proof}
	Let us start with the complex case, and prove the implication $(3) \implies (1)$. Writing the Schmidt decompositions
	\begin{align*}
		\ket \phi &= \sum_i \sqrt{p_i} \ket{a_i} \otimes \ket{b_i}\\
		\ket \chi &= \sum_j \sqrt{q_j} \ket{c_j} \otimes \ket{d_j},
	\end{align*}
	we have 
	$$\mathcal E ^{\otimes 2}(\ket{\phi}\bra{\chi}) = \sum_{ij} \sqrt{p_iq_j} \mathcal E(\ket{a_i}\bra{c_j}) \otimes \mathcal E(\ket{b_i}\bra{d_j}).$$
	Since $\mathcal E$ is an isometry, one recognizes above the Schmidt decomposition of the left hand side, so 
	$$\| \mathcal E ^{\otimes 2}(\ket{\phi}\bra{\chi}) \|_{\ell_2^{n} \otimes_\pi \ell_2^{n}} = \sum_{ij}\sqrt{p_iq_j} = \left(\sum_i \sqrt{p_i}\right)\left(\sum_j \sqrt{q_j}\right),$$
	which implies the claim, since at least one of the last two factors is strictly larger than $1$. 
	
	Let us now move to the implication $(1) \implies (2)$. Consider any four unit vectors $x,y,x',y' \in \mathbb C^{d}$, and pick orthogonal unit vectors $x_\perp, y_\perp,x'_\perp,y'_\perp \in \mathbb C^{d}$ (we assume here $d\geq 2$). Define, for $k \in \mathbb N$, the entangled vectors
	\begin{align*}
		\ket{\phi_k} &= \sqrt{1-\frac 1 k} \ket{x} \otimes \ket{x'} + \sqrt{\frac 1 k} \ket{x_\perp} \otimes \ket{x'_\perp} \\
		\ket{\chi_k} &= \sqrt{1-\frac 1 k} \ket{y} \otimes \ket{y'} + \sqrt{\frac 1 k} \ket{y_\perp} \otimes \ket{y'_\perp}.
	\end{align*}
	Using the hypothesis that $\mathcal E$ is a $\mathbb C$-perfect tester, we have $\|\mathcal E ^{\otimes 2}(\ket{\phi_k}\bra{\chi_k})\|_{\ell_2^{n} \otimes_\pi \ell_2^{n}}>1$, hence
	$$\|\mathcal E(\ket{x}\bra{y})\|_{\ell_2^{n}} \|\mathcal E(\ket{x'}\bra{y'})\|_{\ell_2^{n}} = \lim_{k\to \infty} \|\mathcal E ^{\otimes 2}(\ket{\phi_k}\bra{\chi_k})\|_{\ell_2^{n} \otimes_\pi \ell_2^{n}} \geq 1.$$
	Note that since $\mathcal E$ is a $\mathbb C$-tester, we also have $\|\mathcal E(\ket{x}\bra{y})\|_{\ell_2^{n}}\leq 1$ and $\|\mathcal E(\ket{x'}\bra{y'})\|_{\ell_2^{n}} \leq 1$, so actually both norms are equal to $1$, proving the claim that $\mathcal E$ preserves the Euclidean norms of unit rank operators. 
	
	Finally, the implication $(2) \implies (3)$ follows from Lemma \ref{lem:product-numerical-range} applied to $\mathcal E^*\mathcal E$, in which we identify $\mathcal M_d(\mathbb C) \cong \mathbb C^d \otimes \mathbb C^d$.
	
	The equivalence in the real case can be proven in a similar manner.
\end{proof}

\begin{remark}
Amongst the examples of testers presented in Section \ref{sec:important-examples}, we see that the ones defined from matrix bases are $\mathbb C$-perfect testers, while the ones defined from $2$-designs are only $\mathbb R$-perfect testers. 

Note also that the two conditions above corresponding to the real case are not equivalent to the stronger condition

(3') $\mathcal F$ is an isometry $S_2^{d,sa} \to \ell_2^{n}$.

\noindent Indeed, there exist $\mathbb R$-linear maps preserving the Euclidean norm of any unit rank self-adjoint matrix $\kb{x}{x}$ which are not isometries. An example is the SIC POVM tester $\mathcal{S}$ described in Section \ref{sec:important-examples}. Indeed, from equation \eqref{eq:norm-S-X}, we have $\|\mathcal S(\kb{x}{x})\|_2 = 1$ for any unit vector $x \in \mathbb C^d$. However, we have 
$$\|\mathcal S(I)\|_2 = \left\| \sqrt\frac{d+1}{2d} I \right\|_2 = \sqrt{\frac{d+1}{2}} < \sqrt{d}= \|I\|_2.$$
\end{remark}

The statement of the following lemma is very similar to \cite[Lemma 2.1]{johnston2011characterizing}. Also, its proof can be seen to follow from \cite[Proposition 3]{puchala2011product}.

\begin{lemma}\label{lem:product-numerical-range}
	Let $A \in M_{d_1d_2}(\mathbb C)$ be such that for all unit vectors $x \in \mathbb C^{d_1}$, $y \in \mathbb C^{d_2}$, 
	\begin{equation}\label{eq:product-numerical-range-general}
	\bra {x \otimes y} A \ket{x \otimes y} = 1.
	\end{equation}
	Then, $A = I_{d_1d_2}$. 
\end{lemma}
\begin{proof}
Re-write \eqref{eq:product-numerical-range-general} as 
	$$\mathrm{Tr} \left( (A-I) \ket{x}\bra{x} \otimes \ket{y}\bra{y} \right) = 0.$$
	Since the $\mathbb C$-linear span of $\ket{x}\bra{x} \otimes \ket{y}\bra{y}$ is the whole matrix algebra $\mathcal M_{d_1d_2}(\mathbb C)$, we get $A-I = 0$. 
\end{proof}

\subsection{Test operators associated to perfect testers} \hfill\smallskip

We have seen in the previous subsection that $\mathbb C$-perfect testers are precisely those for which the map $\mathcal E$ is an isometry $S_2^d \to \ell_2^n$. 

\begin{theorem}
Let $\mathcal E :S_1^d \to \ell_2^{d^2}$ be a $\mathbb C$-perfect tester. The test operator of $\mathcal E$ is then $T_E = F$, the flip operator. Hence, by Proposition \ref{prop:equal-T-equivalent}, $\mathcal E$ is equivalent to the realignment tester $\mathcal R$, in the sense of Definition \ref{def:equivalent-testers}.
\end{theorem}
\begin{proof}
    Since $\mathcal E$ is an isometry, we have (identifying $\mathcal M_d(\mathbb C)$ with $\mathbb C^{d^2}$ as vector spaces) $\Theta_E = \mathcal E^*\mathcal E = I_{d^2}$, and thus $T_E = F$.
\end{proof}

Let us now analyze the case of $\mathbb R$-perfect testers. 
Having characterized them in Theorem \ref{th:perfect-testers}, we would now like to re-express what this means at the level of the associated test operator. In other words: given an $\mathbb R$-linear map $\mathcal E : S_1^{d,sa} \to \ell_2^{n}$ such that the norm $\|\mathcal E \|_{S_1^{d,sa} \to \ell_2^{n}} = 1$ is attained at all the extremal points of the unit ball of $S_1^{d,sa}$, how is its associated test operator $T_E$ characterized? We state the answer bellow.

\begin{theorem} \label{th:perfect-tester-T}
If an $\mathbb R$-linear map $\mathcal E : S_1^{d,sa} \to \ell_2^{n}$ is an $\mathbb R$-perfect tester, then its associated test operator $T_E$ is of the form
\[ T_E=\frac{I+F}{2}+T'_E, \] 
with $T'_E$ orthogonal to $(I+F)/2$.
\end{theorem}

Let $\mathcal E : S_1^{d,sa} \to \ell_2^{n}$ be such that, for all unit vector $x\in\mathbb C^d$, $\|\mathcal E (\kb{x}{x})\|=1$. As we have already seen before, this means that, for all unit vector $x\in\mathbb C^d$, $\bra{x\otimes x}T_E\ket{x\otimes x}=1$. The statement in Theorem \ref{th:perfect-tester-T} is thus an immediate consequence of Lemma \ref{lem:perfect-tester-T} below.

\begin{lemma} \label{lem:perfect-tester-T}
Let $T$ be an operator on $\mathbb{C}^d\otimes\mathbb{C}^d$ such that, for any unit vector $x\in\mathbb{C}^d$, $\bra{x\otimes x}T\ket{x\otimes x}=1$. Then, 
\[  T=\frac{I+F}{2}+T', \] 
with $T'$ orthogonal to $(I+F)/2$.
\end{lemma}

\begin{proof}
Decompose $T$ into its symmetric and anti-symmetric parts as
\[  T = \Pi_S T \Pi_S + \Pi_S T \Pi_A + \Pi_A T \Pi_S + \Pi_A T \Pi_A =: \Pi_S T \Pi_S + T' , \]
where $\Pi_S=(1+F)/2$ and $\Pi_A=(1-F)/2$ are the projectors onto the symmetric and anti-symmetric subspaces of $\mathbb{C}^d\otimes\mathbb{C}^d$. Now, for any unit vector $x\in\mathbb{C}^d$, $x\otimes x$ belongs to the symmetric subspace of $\mathbb{C}^d\otimes\mathbb{C}^d$. Hence, $\bra{x\otimes x}T'\ket{x\otimes x}=0$ and $\bra{x\otimes x}\Pi_S\ket{x\otimes x}=1$, so that by assumption $\bra{x\otimes x}\Pi_S T \Pi_S - \Pi_S\ket{x\otimes x}=0$. And since the span of the $x\otimes x$'s is actually the whole symmetric subspace of $\mathbb{C}^d\otimes\mathbb{C}^d$, this means that $\Pi_S T \Pi_S = \Pi_S$, as announced.
\end{proof}

As examples of operators $T'$ satisfying the condition of Theorem \ref{th:perfect-tester-T}, we have any operator supported on the anti-symmetric subspace of $\mathbb{C}^d\otimes\mathbb{C}^d$, in particular any multiple of $(I-F)/2$, the projector on this subspace. This gives the following family of corresponding operators $T$:
\begin{equation} \label{eq:T_delta} T_{\delta} := \frac{I+F}{2} + \delta \frac{I-F}{2} = \frac{1}{2}\left( (1+\delta)I + (1-\delta)F \right). \end{equation}
In order for these to actually be a test operator, we further need to impose that
\[ \forall\ X\in\mathcal M_d(\mathbb C), \quad \mathrm{Tr}(T_{\delta}^* X\otimes X^*) \geq 0 . \]
Now, given $X\in\mathcal M_d(\mathbb C)$,
\[ \mathrm{Tr}(T_{\delta}^* X\otimes X^*) = \frac{1}{2}\left( (1+\delta)|\mathrm{Tr}(X)|^2 + (1-\delta)\mathrm{Tr}(|X|^2) \right) , \]
and the latter quantity is always non negative if and only if $-1 \leq \delta \leq 1$. 

Note that the operator Schmidt rank of $T_{\delta}$, $-1 \leq \delta < 1$, is simply the rank of 
\[ \frac{1}{2}\left((1+\delta)d\kb{\psi}{\psi} + (1-\delta) F\right) , \] i.e.~$d^2$. This means that a tester $\mathcal{E}$ having as associated test operator $T_{\delta}$, $-1 \leq \delta < 1$, needs to be composed of at least $d^2$ operators $\{E_1,\ldots,E_n\}$.

Testers having a test operator of the form described by equation \eqref{eq:T_delta} play a central role in our paper. The two main examples of testers that we are considering, namely the realignment tester $\mathcal R$ and the SIC POVM tester $\mathcal S$, defined in Section \ref{sec:important-examples}, enter in this category. More specifically, we have $T_R=T_{-1}$ and $T_S=T_0$. In Section \ref{sec:symmetric-testers} we explain a general approach to construct testers of this kind.

\section{Construction of testers from symmetric families of operators}\label{sec:symmetric-testers}

In this section we investigate what are the conditions on a set of operators $\{E_k\}_{k=1}^n$ on $\mathbb{C}^d$ so that its associated operator $T_E$ on $\mathbb{C}^d \otimes \mathbb{C}^d$ is a linear combination of $I$ and $F$, i.e.
\[  \sum_{k=1}^n E_k\otimes E_k^* = \alpha F + \beta I . \]
In \cite{appleby} this problem was studied in the case where the $E_k$'s are Hermitian, but the result obtained there easily generalises to the non-Hermitian case. The equivalent of \cite[Theorem 1]{appleby} reads as follows.

\begin{theorem}\label{thequiv}
Let $\{E_k\}_{k=1}^{d^2}$ be a basis of operators on $\mathbb{C}^d$. Then, the following statements are equivalent
\begin{align}
\label{eq1}
& \sum_{k=1}^{d^2} E_k\otimes E_k^*=(\beta+\alpha)\frac{I+F}{2}+(\beta-\alpha)\frac{I-F}{2} = \alpha F + \beta I , \\ 
\label{eq2}
& \forall\ 1\leq k,l\leq d^2, \quad \mathrm{Tr}(E_k^*E_l) = \alpha\delta_{kl}+\gamma\mathrm{Tr}(E_k^*)\mathrm{Tr}(E_l) .
\end{align}
In this case, we have $\alpha>0$, $\alpha+d\beta>0$ and $\gamma=\beta/(\alpha+d\beta)$.
\end{theorem}

Explicitly, the constants $\alpha$ and $\beta$ are thus given by the following formulas
 \begin{align*}
 \alpha & = \frac{1}{d^3-d}\left( d\sum_{k=1}^{d^2}\mathrm{Tr}(E_k^*E_k) - \sum_{k=1}^{d^2}\mathrm{Tr}(E_k^*)\mathrm{Tr}(E_k) \right) ,\\
 \beta & = \frac{1}{d^3-d} \left( -\sum_{k=1}^{d^2} \mathrm{Tr}(E_k^*E_k) + d\sum_{k=1}^{d^2} \mathrm{Tr}(E_k^*)\mathrm{Tr}(E_k) \right) .
\end{align*}

A family of operators $\{E_k\}_{k=1}^{d^2}$ satisfying the equivalent conditions given in Theorem \ref{thequiv} can be seen as a generalized minimal $2$-design. As we will show next, this framework actually encompasses the three examples that we mentioned in Section \ref{sec:important-examples}. 

In the case of the realignment map $\mathcal R$, where the associated set of operators is the basis of matrix units $R=\{R_{ij}\}_{i,j=1}^d=\{|i\rangle\langle j|\}_{i,j=1}^d$, we have 
\begin{align*}
& \forall\ 1\leq i,j\leq d, \quad \mathrm{Tr}(R_{ij})=\delta_{ij} \, , \\
& \forall\ 1\leq i,j,k,l\leq d, \quad \mathrm{Tr}(R_{ij}^* R_{lk})=\delta_{il}\delta_{jk}  \, . 
\end{align*}
Consequently, the parameters from Theorem \ref{thequiv} are $\alpha=1$ and $\beta=\gamma=0$, so that $T_R=F$.

In the case of the map $\mathcal{G}$, where the associated set of operators is a canonical matrix basis $G=\{G_{k}\}_{k=1}^{d^2}$, we have
\begin{align*}
& \mathrm{Tr}(G_1)=\sqrt{d} \quad \text{and} \quad \forall\ 2\leq k\leq d^2, \quad \mathrm{Tr}(G_k)=0 \, , \\
& \forall\ 1\leq k,l\leq d^2, \quad \mathrm{Tr}(G_k^*G_l)=\delta_{kl} \, .
\end{align*}
Consequently, as in the previous example, the parameters from Theorem \ref{thequiv} are $\alpha=1$ and $\beta=\gamma=0$, so that $T_G=F$.

In the case of the SIC POVM map $\mathcal S$,  where the associated set of operators is a renormalized symmetric $2$-design $S=\{S_k\}_{k=1}^{d^2}=\{\sigma |x_k\rangle\langle x_k|\}_{k=1}^{d^2}$, with $\sigma=\sqrt{(d+1)/(2d)}$, we have 
\begin{align*}
& \forall\ 1\leq k\leq d^2, \quad \mathrm{Tr}(S_k)=\sqrt{\frac{d+1}{2d}} \, , \\
& \forall\ 1\leq k,l\leq d^2, \quad \mathrm{Tr}(S_k^*S_l)=\frac{d+1}{2d} \left( \delta_{kl} + \frac{1}{d+1}(1-\delta_{kl}) \right) \, .
\end{align*}
Consequently, the parameters from Theorem \ref{thequiv} are $\alpha=\beta=1/2$ and $\gamma=1/{(d+1)}$, so that $T_S=(I+F)/2$ .

The latter example actually generalizes the situation where the $E_k$'s are coming from a so-called \emph{non-degenerate symmetric family of operators}, as defined in \cite{jivulescu2017symmetric}. This means that 
\begin{align}\nonumber
& \forall\ 1\leq k\leq d^2,\ \mathrm{Tr}(E_k)= t , \\\label{sym-fam}
& \forall\ 1\leq k,l\leq d^2,\ \mathrm{Tr}(E_k^*E_l)=a\delta_{kl}+b(1-\delta_{kl}) .
\end{align} 
In this case, equation \eqref{eq2} holds with $\alpha=a-b$ and $\gamma=b/|t|^2$, so that equation \eqref{eq1} holds with $\alpha=a-b$ and $\beta=(a-b)b/(|t|^2-db)$. Note that the operators composing the SIC POVM map $\mathcal S$ do satisfy equations \eqref{sym-fam}. 

In \cite{appleby2016} sets of operators satisfying equations \eqref{sym-fam} are studied as well. There, only Hermitian families are considered, but with any number $n$ (not necessarily equal to $d^2$) of operators. Such families are dubbed \emph{conical $2$-designs}. And it is shown that the class of conical $2$-designs includes several sub-classes of operators which are relevant for quantum information, such as arbitrary rank symmetric informationally complete measurements (SIMs) or full sets of arbitrary rank mutually unbiased measurements (MUMs).

\begin{lemma} \label{lem:norm-symmetric}
    Let $\mathcal E:S_1^d \to \ell_2^n$ be a $\mathbb C$-linear map such that $T_{E} = \alpha F + \beta I$, with $\alpha\geq 0$ and $\beta\geq -\alpha/d$. Then 
    $$\|\mathcal E\|_{S_1^d \to \ell_2^n} =  \begin{cases} \sqrt{\alpha + \beta} \quad \text{if } \beta\geq 0 \\
    \sqrt{\alpha} \quad \text{if } \beta< 0 \end{cases}.$$
\end{lemma}

\begin{proof}
    We have
    $$\|\mathcal E\|_{S_1^d \to \ell_2^n}^2 = \sup_{\|X\|_1 \leq 1} \mathrm{Tr}( T_{E}^* X\otimes X^* ) = \sup_{\|X\|_1 \leq 1} \alpha \|X\|_2^2 + \beta |\Tr X|^2 .$$
    If $\beta \geq 0$, the above supremum is equal to $\alpha+\beta$. Indeed, for any $X$, $\|X\|_2 \leq \|X\|_1$ and $|\Tr X| \leq \|X\|_1$, with both inequalities being saturated by rank one projections. While if $\beta<0$, the above supremum is equal to $\alpha$. This is because, for any $X$, $\|X\|_2 \leq \|X\|_1$ and $|\Tr X| \geq 0$, with both inequalities being saturated by rank one operators of the form $\ket{x}\bra{y}$ for orthogonal unit vectors $x,y$.
\end{proof}

As a consequence of Lemma \ref{lem:norm-symmetric} we have that, if $\alpha,\beta\geq 0$ and $\alpha+\beta=1$, or $\alpha=1$ and $-1/d\leq \beta<0$, then any $\mathbb C$-linear map $\mathcal E:S_1^d \to \ell_2^n$ such that $T_{E} = \alpha F + \beta I$ is a tester. Conversely, if $T=\alpha F + \beta I$ with $\alpha,\beta$ satisfying the above conditions, then we can exhibit operators $\{E_k\}_{k=1}^{d^2}$ such that $T$ is the test operator associated to the corresponding map $\mathcal E$, i.e.~ such that $T=\sum_{k=1}^{d^2} E_k\otimes E_k^*$. The construction follows the strategy described in Section \ref{sec:test-operator}. Set $\Theta=(TF)^{\Gamma}$, i.e.
$$ \Theta=\alpha I+\beta d\kb{\psi}{\psi} . $$
$\Theta$ can be diagonalized as
$$ \Theta = (\alpha+\beta d)\kb{x_1}{x_1} + \alpha\kb{x_2}{x_2} + \cdots + \alpha\kb{x_{d^2}}{x_{d^2}} , $$
where $x_1=\psi$ and $x_2,\ldots,x_{d^2}$ are such that $\{x_1,\ldots,x_{d^2}\}$ forms an orthonormal basis of $\mathbb C^d \otimes \mathbb C^d$. Defining $\ket{e_1}=\sqrt{\alpha+\beta d}\ket{x_1}$ and $\ket{e_k}=\sqrt{\alpha}\ket{x_k}$ for $2\leq k\leq d^2$, we then have
$$ T=\sum_{k=1}^{d^2} E_k\otimes E_k^* , $$
where the $E_k$'s are the matrix versions of the $e_k$'s. Concretely, this means that $E_1=\sqrt{\alpha+\beta d}G_1$ and $E_k=\sqrt{\alpha}G_k$ for $2\leq k\leq d^2$, with $\{G_1,\ldots,G_{d^2}\}$ a canonical basis of $\mathcal M_d(\mathbb C)$ (i.e.~$G_1$ is a multiple of the identity and $G_2,\ldots,G_{d^2}$ are traceless). It is interesting to note that, in the case $\alpha=1,\beta=0$ this canonical construction yields the map $\mathcal G$, while in the case $\alpha=\beta=1/2$ it provides an alternative map having the same test operator as the map $\mathcal S$ (i.e.~being an equivalent tester to $\mathcal S$).

\section{Entanglement detection of bipartite pure states by symmetric testers}
\label{sec:pure-states}

In this section we compute the norm of the action of a given symmetric tester, as defined in Section \ref{sec:symmetric-testers}, on an arbitrary pure bipartite quantum state. We focus next on the realignment and SIC POVM testers, establishing an equality between the corresponding norms which was conjectured in \cite{shang2018enhanced}.

Let $\mathcal{E}:S_1^d\to\ell_2^n$ be a linear map, defined by
$$ \mathcal E: X\in \mathcal{M}_d(\mathbb{C}) \mapsto \sum_{k=1}^n \mathrm{Tr}(E_k^* X) \ket{k} \in \mathbb{C}^n . $$
Denote by $T_E$ its associated test operator, as defined by equation \eqref{eq:T_E}, which we assume to be symmetric, in the sense of equation \eqref{sym-fam}, i.e.
\[ T_E = \alpha F+\beta I ,\]
for some parameters $\alpha\geq 0$ and $\beta\geq-\alpha/d$. Note that, for now, we do not ask that $\mathcal E$ should be normalized to be an entanglement tester. 

Consider an arbitrary bipartite unit vector $\varphi\in\mathbb{C}^d \otimes \mathbb{C}^d$, with Schmidt decomposition \begin{equation}\label{eq:Schmidt-decomp-phi} \ket{\varphi}=\sum_{i=1}^r\sqrt{\lambda_i}\ket{e_if_i},
\end{equation}
where $\lambda_1,\ldots,\lambda_r>0$ are such that $\sum_{i=1}^r \lambda_i=1$ and $\{e_1,\ldots,e_r\}$, $\{f_1,\ldots,f_r\}$ are orthonormal families in $\mathbb{C}^d$.

\begin{proposition}
	Let $\mathcal E:S_1^d\to\ell_2^n$ be a linear map as above, which is symmetric in the sense of equation \eqref{sym-fam} with corresponding parameters $(\alpha, \beta)$. Then, for any bipartite unit vector $\varphi\in\mathbb{C}^d \otimes \mathbb{C}^d$ as above, we have
	$$\left\| \mathcal{E}^{\otimes 2}(\kb{\varphi}{\varphi}) \right\|_1 = \alpha + \beta + 2\alpha \sum_{i<j} \sqrt{\lambda_i \lambda_j}.$$
\end{proposition}

\begin{proof}
Start from the Schmidt decomposition \eqref{eq:Schmidt-decomp-phi}, and set $\ket{u_{ij}}:= \sum_{k=1}^n  \bra{e_j}E_k^*\ket{e_i} \ket{k}$, $\ket{v_{ij}}:= \sum_{k=1}^n \bra{f_j}E_k^*\ket{f_i} \ket{k}$, $1\leq i,j\leq r$. We then have, viewing $\mathcal{E}^{\otimes 2}(\kb{\varphi}{\varphi})$ as belonging to $\mathcal M_n(\mathbb C)$ rather than $\mathbb C^n\otimes \mathbb C^n$,
\[ \mathcal{E}^{\otimes 2}(\kb{\varphi}{\varphi}) = \sum_{i,j=1}^{r} \sqrt{\lambda_i\lambda_j} \ket{u_{ij}}\bra{\bar{v}_{ij}}
.\]

Let us begin with considering the situation where the $f_i$'s are equal to the $\bar{e}_i$'s in \eqref{eq:Schmidt-decomp-phi}. This implies that the $\bar{v}_{ij}$'s are equal to the $u_{ij}$'s, and therefore that $\mathcal{E}^{\otimes 2}(\kb{\varphi}{\varphi})$ is positive semidefinite. Hence,
\[ \left\| \mathcal{E}^{\otimes 2}(\kb{\varphi}{\varphi}) \right\|_{1} = \operatorname{Tr} \left( \mathcal{E}^{\otimes 2}(\kb{\varphi}{\varphi}) \right)= \sum_{i,j=1}^{r} \sqrt{\lambda_i\lambda_j} \|u_{ij}\|^2.\]
Observe next that, for all $1\leq i,j\leq r$, the symmetry of the map $\mathcal E$ implies that
\begin{equation}\label{eq:Gram-u}
\bk{u_{ij}}{u_{i'j'}} = \bra{e_ie_{j'}}T_E\ket{e_je_{i'}} = \alpha \delta_{ii'}\delta_{jj'} + \beta \delta_{ij} \delta_{i'j'}.
\end{equation}
Therefore, $\|u_{ij}\|^2 = \alpha + \beta \delta_{ij}$ and the conclusion follows for the special case where $u_{ij} = v_{ij}$. 

To conclude, we are just left with understanding what happens when the $f_i$'s are not equal to the $\bar{e}_i$'s. Note that the vectors $v_{ij}$ also satisfy equation \eqref{eq:Gram-u}, so the families $\{u_{ij}\}_{i,j=1}^r$ and $\{v_{ij}\}_{i,j=1}^r$ have the same Gram matrix. Thus, there exists a unitary operator $W_E$ on $\mathbb{C}^n$ mapping the $u_{ij}$'s to the $\bar{v}_{ij}$'s. Hence, using the invariance of $\|\cdot\|_1$ under multiplication by a unitary,
\[ \left\| \sum_{i,j=1}^{r} \sqrt{\lambda_i\lambda_j} \kb{u_{ij}}{\bar{v}_{ij}} \right\|_1 = \left\| \left( \sum_{i,j=1}^{r} \sqrt{\lambda_i\lambda_j} \kb{u_{ij}}{u_{ij}} \right) W_E^* \right\|_1 = \left\| \sum_{i,j=1}^{r} \sqrt{\lambda_i\lambda_j} \kb{u_{ij}}{u_{ij}} \right\|_1,\]
finishing the proof.
\end{proof}

Let us consider now the special cases of the realignment tester $\mathcal R$ and the SIC POVM tester $\mathcal S$. These maps are symmetric testers, with respective parameters $(\alpha, \beta) = (1,0)$ and $(\alpha, \beta) = (1/2, 1/2)$. In both cases, $\alpha+\beta=1$ and $\alpha > 0$. Hence, the tester provides a necessary and sufficient condition for separability of bipartite pure states, namely: for any bipartite pure state $\varphi$, $\left\| \mathcal{E}^{\otimes 2}(\kb{\varphi}{\varphi}) \right\|_1\leq 1$ if and only if $\varphi$ is separable. More precisely, we have 
\begin{align*}
\left\| \mathcal{R}^{\otimes 2}(\kb{\varphi}{\varphi}) \right\|_1 &= \sum_{1\leq i,j\leq r} \sqrt{\lambda_i\lambda_j} = 1 + \sum_{1\leq i\neq j\leq r} \sqrt{\lambda_i\lambda_j}, \\
\left\| \mathcal{S}^{\otimes 2}(\kb{\varphi}{\varphi}) \right\|_1 &= \sum_{1 \leq i \leq j \leq r} \sqrt{\lambda_i \lambda_j} = 1+\sum_{1 \leq i < j \leq r} \sqrt{\lambda_i \lambda_j}.
\end{align*}
Note that we also know from \cite[Proof of Theorem 0.1]{palazuelos2014largest} that
$$ \| \kb{\varphi}{\varphi}\|_{S_1^d \otimes_\pi S_1^d} = \sum_{1\leq i,j\leq r} \sqrt{\lambda_i\lambda_j} ,$$
so that we actually have
$$\left\| \mathcal{R}^{\otimes 2}(\kb{\varphi}{\varphi}) \right\|_1 = \| \kb{\varphi}{\varphi}\|_{S_1^d \otimes_\pi S_1^d}.$$

Moreover, the realignment map is always `better' than the SIC POVM map on pure states, in the sense that: for any pure state $\varphi$,
\[ \left\| \mathcal{R}^{\otimes 2}(\kb{\varphi}{\varphi}) \right\|_1 \geq \left\| \mathcal{S}^{\otimes 2}(\kb{\varphi}{\varphi}) \right\|_1,\]
with strict inequality as soon as $\varphi$ is entangled. More precisely: for any pure state $\varphi$,
\begin{equation}\label{eq:pure-states}
    \left\| \mathcal{R}^{\otimes 2}(\kb{\varphi}{\varphi}) \right\|_1 -1 = 2\left( \left\| \mathcal{S}^{\otimes 2}(\kb{\varphi}{\varphi}) \right\|_1 -1 \right).
\end{equation}
This proves the conjectured equality in \cite[equation (21)]{shang2018enhanced}.

\section{Entanglement detection of bipartite isotropic and Werner states by the realignment and the SIC POVM testers}
\label{sec:symmetric-states}

The goal here is to determine when the realignment and SIC POVM testers detect the entanglement of isotropic and Werner states. These are defined, respectively, as
\[ \tau_{\mu} := \mu\kb{\psi}{\psi} + (1-\mu)\frac{I}{d^2} \, , \ 0\leq\mu\leq 1  , \]
\[ \sigma_{\mu} := \mu\frac{I+F}{d(d+1)} + (1-\mu)\frac{I-F}{d(d-1)} \, , \ 0\leq\mu\leq 1  . \]
For that, let us start with computing the action of $\mathcal{R}^{\otimes 2}$ and $\mathcal{S}^{\otimes 2}$ on $I$, $\kb{\psi}{\psi}$ and $F$.

First of all,
\[ \mathcal{R}^{\otimes 2}(I) = \sum_{i,j,k,l=1}^d \mathrm{Tr}(\kb{jl}{ik}) \kb{ij}{kl} = \sum_{i,k=1}^d \kb{ii}{kk} = d \kb{\psi}{\psi} , \]
\[ \mathcal{S}^{\otimes 2}(I) = \frac{d+1}{2d} \sum_{i,j=1}^{d^2} \mathrm{Tr}(\kb{x_ix_j}{x_ix_j}) \kb{i}{j} = \frac{d+1}{2d} \sum_{i,j=1}^{d^2} \kb{i}{j} = \frac{d+1}{2d} J . \] 
Then observe that, for any operators $A,B$ on $\mathbb{C}^d$, $\mathrm{Tr}(A\otimes B \kb{\psi}{\psi})= \mathrm{Tr}(AB^\top)/d$. Hence,
\[ \mathcal{R}^{\otimes 2}(\kb{\psi}{\psi}) = \frac{1}{d} \sum_{i,j,k,l=1}^d \mathrm{Tr}(\kb{j}{i}\kb{k}{l}) \kb{ij}{kl} = \frac{1}{d} \sum_{i,j=1}^d \kb{ij}{ij} = \frac{I}{d} , \]
\begin{align*} 
\mathcal{S}^{\otimes 2}(\kb{\psi}{\psi}) & = \frac{d+1}{2d^2} \sum_{i,j=1}^{d^2} \mathrm{Tr}(\kb{x_i}{x_i}\kb{x_j}{x_j}) \kb{i}{j} \\
& =  \frac{d+1}{2d^2} \left( \left(1-\frac{1}{d+1}\right) \sum_{i=1}^{d^2} \kb{i}{i} + \frac{1}{d+1} \sum_{i,j=1}^{d^2} \kb{i}{j} \right) \\
& = \frac{1}{2d} \left( I + \frac{1}{d} J \right)  . 
\end{align*}
Finally observe that, for any operators $A,B$ on $\mathbb{C}^d$, $\mathrm{Tr}(A\otimes B F)= \mathrm{Tr}(AB)$. Hence,
\[ \mathcal{R}^{\otimes 2}(F) = \sum_{i,j,k,l=1}^d \mathrm{Tr}(\kb{j}{i}\kb{l}{k}) \kb{ij}{kl} = \sum_{i,j=1}^d \kb{ij}{ji} = F , \]
\[ \mathcal{S}^{\otimes 2}(F) = \frac{d+1}{2d} \sum_{i,j=1}^{d^2} \mathrm{Tr}(\kb{x_i}{x_i}\kb{x_j}{x_j}) \kb{i}{j} = \frac{1}{2} \left( I + \frac{1}{d} J \right) . \]

\subsection{Isotropic states} \hfill\smallskip

With these preliminary computations at hand, let us start with understanding the entanglement detection of isotropic states. For any $0\leq \mu\leq 1$, we have
\begin{align*} 
& \mathcal{R}^{\otimes 2}(\tau_{\mu}) = \frac{1}{d} \left(\mu I + (1-\mu)\kb{\psi}{\psi} \right) , \\
& \mathcal{S}^{\otimes 2}(\tau_{\mu}) = \frac{1}{2d} \left(\mu I + \frac{d+1-\mu}{d^2}J \right) . 
\end{align*}
Since $\mathcal{R}^{\otimes 2}(\tau_{\mu})$ and $\mathcal{S}^{\otimes 2}(\tau_{\mu})$ are positive semidefinite, we then simply have
\begin{align*}
& \left\| \mathcal{R}^{\otimes 2}(\tau_{\mu}) \right\|_1 = \mathrm{Tr}\left(\mathcal{R}^{\otimes 2}(\tau_{\mu})\right) = \frac{1}{d}((d^2-1)\mu + 1) , \\
& \left\| \mathcal{S}^{\otimes 2}(\tau_{\mu}) \right\|_1 = \mathrm{Tr}\left(\mathcal{S}^{\otimes 2}(\tau_{\mu})\right) = \frac{d+1}{2d}((d-1)\mu +1) .
\end{align*}
Hence,
\[ \left\| \mathcal{R}^{\otimes 2}(\tau_{\mu}) \right\|_1 >1 \ \Longleftrightarrow \ \left\| \mathcal{S}^{\otimes 2}(\tau_{\mu}) \right\|_1 >1 \ \Longleftrightarrow \ \mu>\frac{1}{d+1} . \]
As a comparison, we know that we also have $\tau_{\mu}$ entangled iff $\mu>1/(d+1)$. So both the realignment and the SIC POVM maps detect all entangled isotropic states.

What is more, we know from \cite[Theorem 11]{rudolph2005further} that
\[ \|\tau_{\mu}\|_{S_1^d\otimes_{\pi} S_1^d} = \begin{cases} 1 \text{ if } \mu\leq 1/(d+1) \\
((d^2-1)\mu +1)/d \text{ if } \mu> 1/(d+1) \end{cases} .\]
So we actually have that, for any entangled isotropic state $\tau_{\mu}$,
\[  \left\| \mathcal{R}^{\otimes 2}(\tau_{\mu}) \right\|_1 = \|\tau_{\mu}\|_{S_1^d\otimes_{\pi} S_1^d} .\]

\subsection{Werner states} \hfill\smallskip

Let us now turn to understanding the entanglement detection of Werner states. For any $0\leq \mu\leq 1$, we have
\begin{align*} 
\mathcal{R}^{\otimes 2}(\sigma_{\mu}) & = \frac{\mu}{d(d+1)}(d\kb{\psi}{\psi}+F) + \frac{1-\mu}{d(d-1)}(d\kb{\psi}{\psi}-F) \\
& = \frac{d+1-2\mu}{d^2-1}\kb{\psi}{\psi} + \frac{2\mu d-d-1}{d(d^2-1)} F \, , \\
\mathcal{S}^{\otimes 2}(\sigma_{\mu}) & = \frac{\mu}{2d(d+1)}\left(\frac{d+2}{d}J+I\right) + \frac{1-\mu}{2d(d-1)}(J-I) \\
& = \frac{2\mu d-d-1}{2d(d^2-1)} I + \frac{d^2+d-2\mu}{2d^2(d^2-1)} . 
\end{align*}
We then have to distinguish two cases. If $\mu\geq (d+1)/2d$, then
\begin{align*}
& \left\| \mathcal{R}^{\otimes 2}(\sigma_{\mu}) \right\|_1 = \frac{d+1-2\mu}{d^2-1} + \frac{2\mu d-d-1}{d(d^2-1)} \times d^2 = 2\mu -1 , \\
& \left\| \mathcal{S}^{\otimes 2}(\sigma_{\mu}) \right\|_1 = \frac{2\mu d-d-1}{2d(d^2-1)} \times d^2 + \frac{d^2+d-2\mu}{2d^2(d^2-1)} \times d^2 = \mu ,
\end{align*}
which are both smaller than $1$. While if $\mu\leq (d+1)/2d$, then
\begin{align*}
& \left\| \mathcal{R}^{\otimes 2}(\sigma_{\mu}) \right\|_1 = \frac{d+1-2\mu}{d^2-1} + \frac{d+1-2\mu d}{d(d^2-1)} \times (d^2-2) = \frac{d+2}{d}-2\mu , \\
& \left\| \mathcal{S}^{\otimes 2}(\sigma_{\mu}) \right\|_1 = \frac{d+1-2\mu d}{2d(d^2-1)} \times (d^2-2) + \frac{d^2+d-2\mu}{2d^2(d^2-1)} \times d^2 = \frac{d+1}{d}- \mu ,
\end{align*}
so that 
\[ \left\| \mathcal{R}^{\otimes 2}(\sigma_{\mu}) \right\|_1 >1 \ \Longleftrightarrow \ \left\| \mathcal{S}^{\otimes 2}(\sigma_{\mu}) \right\|_1 >1 \ \Longleftrightarrow \ \mu<\frac{1}{d} . \]
As a comparison, we know that we have $\sigma_{\mu}$ entangled iff $\mu<1/2$. So as soon as $d>2$, both the realignment and the SIC POVM maps do not detect all entangled Werner states (and they perform increasingly poorly as $d$ grows).

What is more, we know from \cite[Theorem 9]{rudolph2005further} that
\[ \|\sigma_{\mu}\|_{S_1^d\otimes_{\pi} S_1^d} = \begin{cases} 1 \text{ if } \mu\geq 1/2 \\
2(1-\mu) \text{ if } \mu< 1/2 \end{cases} .\]
So for $d=2$ we have that, for any entangled Werner state $\sigma_{\mu}$,
\[  \left\| \mathcal{R}^{\otimes 2}(\sigma_{\mu}) \right\|_1 = \|\tau_{\mu}\|_{S_1^2\otimes_{\pi} S_1^2} .\]
But the two norms do not coincide for $d>2$, even in the regime $\mu<1/d$ where the map $\mathcal R$ detects the entanglement of $\sigma_{\mu}$.

As a final comment, let us point out that, for both isotropic and Werner states, the same equality \eqref{eq:pure-states} as the one established for pure states, relating the norms of the realignment and SIC POVM maps, holds:
\[ \forall\ 0\leq\mu\leq 1, \quad \left\| \mathcal{S}^{\otimes 2}(\tau_{\mu}) \right\|_1 = \frac{\left\| \mathcal{R}^{\otimes 2}(\tau_{\mu}) \right\|_1+1}{2} \quad \text{and} \quad \left\| \mathcal{S}^{\otimes 2}(\sigma_{\mu}) \right\|_1 = \frac{\left\| \mathcal{R}^{\otimes 2}(\sigma_{\mu}) \right\|_1+1}{2} . \]

\section{Entanglement detection of bipartite pure states with white noise by the realignment and the SIC POVM testers}
\label{sec:noisy-states}

We now look at states which are the mixture of a pure state and the maximally mixed state, i.e.~given $\varphi\in\mathbb{C}^d\otimes\mathbb{C}^d$ a unit vector,
\[ \rho_{\mu} := \mu\kb{\varphi}{\varphi} + (1-\mu)\frac{I}{d^2} \, , \ 0\leq\mu\leq 1 . \]
We wonder when the realignment and SIC POVM testers detect the entanglement of such states.

Let us write $\varphi$ in its Schmidt decomposition 
\[ \ket{\varphi}=\sum_{i=1}^r\sqrt{\lambda_i}\ket{e_if_i} . \]
Note that we can assume without loss of generality that the $f_i$'s are equal to the $\bar{e}_i$'s, since any re-labelling of basis leaves $I/d^2$ invariant. If this is so, we have shown before, in Section \ref{sec:pure-states}, that
\[ \mathcal{R}^{\otimes 2}(\kb{\varphi}{\varphi}) = \sum_{i,j=1}^r \sqrt{\lambda_i\lambda_j}\kb{u_{ij}}{u_{ij}} \quad \text{and} \quad \mathcal{S}^{\otimes 2}(\kb{\varphi}{\varphi}) = \sum_{i,j=1}^r \sqrt{\lambda_i\lambda_j}\kb{v_{ij}}{v_{ij}} , \]
where for any $1\leq i,j,i',j'\leq r$,
\[ \bk{u_{ij}}{u_{i'j'}} = \delta_{ii'}\delta_{jj'} \quad \text{and} \quad \bk{v_{ij}}{v_{i'j'}} =\frac{ \delta_{ii'}\delta_{jj'} + \delta_{ij}\delta_{i'j'} }{2} . \]
What is more, we have also shown before, in Section \ref{sec:symmetric-states}, that
\[ \mathcal{R}^{\otimes 2}\left(\frac{I}{d^2}\right) = \frac{1}{d}\kb{\psi}{\psi} \quad \text{and} \quad \mathcal{S}^{\otimes 2}\left(\frac{I}{d^2}\right) = \frac{d+1}{2d^3}J . \]
As a consequence we have that, for any $0\leq \mu\leq 1$,
\begin{align*}
& \mathcal{R}^{\otimes 2}(\rho_{\mu}) = \mu \sum_{i,j=1}^r \sqrt{\lambda_i\lambda_j}\kb{u_{ij}}{u_{ij}} + (1-\mu) \frac{1}{d}\kb{\psi}{\psi} , \\ 
& \mathcal{S}^{\otimes 2}(\rho_{\mu}) = \mu \sum_{i,j=1}^r \sqrt{\lambda_i\lambda_j}\kb{v_{ij}}{v_{ij}} + (1-\mu) \frac{d+1}{2d^3}J .
\end{align*}
Since both are positive semidefinite, we then simply have
\begin{align*}
& \left\| \mathcal{R}^{\otimes 2}(\rho_{\mu}) \right\|_1 = \mathrm{Tr}\left( \mathcal{R}^{\otimes 2}(\rho_{\mu}) \right) = \mu (1 + 2f(\varphi)) + \frac{1-\mu}{d} , \\
& \left\| \mathcal{S}^{\otimes 2}(\rho_{\mu}) \right\|_1 = \mathrm{Tr}\left( \mathcal{S}^{\otimes 2}(\rho_{\mu}) \right) = \mu (1 + f(\varphi)) + \frac{(1-\mu)(d+1)}{2d} ,
\end{align*}
where we have set
\[ f(\varphi):\,= \sum_{i<j=1}^r \sqrt{\lambda_i\lambda_j} . \]

From these expressions, it is easy to see that
\[ \left\| \mathcal{R}^{\otimes 2}(\rho_{\mu}) \right\|_1 > 1 \ \Longleftrightarrow \ \left\| \mathcal{S}^{\otimes 2}(\rho_{\mu}) \right\|_1 > 1 \ \Longleftrightarrow \ \mu > \frac{d-1}{(1+2f(\varphi))d-1} . \]
And that also 
\[ \left\| \mathcal{R}^{\otimes 2}(\rho_{\mu}) \right\|_1 > \left\| \mathcal{S}^{\otimes 2}(\rho_{\mu}) \right\|_1 \ \Longleftrightarrow \ \mu > \frac{d-1}{(1+2f(\varphi))d-1} . \]
This means that the realignment and SIC POVM maps detect the entanglement of $\rho_{\mu}$ below the same amount $1-\mu$ of white noise. And in this range of $\mu$, the realignment map is `better' than the SIC POVM, in the sense that
\[ \left\| \mathcal{R}^{\otimes 2}(\rho_{\mu}) \right\|_1 > \left\| \mathcal{S}^{\otimes 2}(\rho_{\mu}) \right\|_1 . \]

As special cases, we recover the previous results on the entanglement detection of pure states ($\mu=1$) and isotropic states ($\varphi=\psi$). Indeed, for any state $\varphi\in\mathbb{C}^d\otimes\mathbb{C}^d$,
\[ 1 \leq 1+2f(\varphi) \leq d , \]
with equality in the first, resp.~second, inequality iff $\varphi$ is separable, resp.~maximally entangled. Hence, for the case $\mu=1$ we have
\[ \frac{d-1}{(1+2f(\varphi))d-1} < 1 \ \Longleftrightarrow \ 1+2f(\varphi) > 1 \ \Longleftrightarrow \ \varphi\ \text{entangled} , \]
while for the case $\varphi=\psi$ we have
\[ \mu > \frac{d-1}{d^2-1} \ \Longleftrightarrow \ \mu > \frac{1}{d+1} \ \Longleftrightarrow \ \tau_{\mu}\ \text{entangled} . \]

Note furthermore that this is yet another class of states for which the same equality \eqref{eq:pure-states} as the one established for pure states, relating the norms of the realignment and SIC POVM maps, holds: 
\[ \forall\ 0\leq\mu\leq 1, \quad \left\| \mathcal{S}^{\otimes 2}(\rho_{\mu}) \right\|_1 = \frac{\left\| \mathcal{R}^{\otimes 2}(\rho_{\mu}) \right\|_1+1}{2} . \]

\section{Entangled states which are detected by the realignment tester are detected by the SIC POVM tester}
\label{sec:conjecture}

In \cite{shang2018enhanced} the following was conjectured: Given an entangled state $\rho$ on $\mathbb C^d \otimes \mathbb C^d$, if its entanglement is detected by the matrix unit tester $\mathcal{R}:S_1^d\to\ell_2^{d^2}$, then it is necessarily detected by the SIC POVM tester $\mathcal{S}:S_1^d\to\ell_2^{d^2}$ as well, i.e.
\begin{equation} \label{eq:conjecture}
\|\mathcal{R}^{\otimes 2}(\rho)\|_{\ell_2^{d^2}\otimes_{\pi}\ell_2^{d^2}} >1 \Longrightarrow \|\mathcal{S}^{\otimes 2}(\rho)\|_{\ell_2^{d^2}\otimes_{\pi}\ell_2^{d^2}} >1 . 
\end{equation}
Here we answer this conjecture in the positive, by showing the following inequality, which clearly implies \eqref{eq:conjecture}. 
\begin{theorem}\label{thm:proof-conjecture}
	For any quantum state $\rho$ on $\mathbb C^d\otimes\mathbb C^d$, we have
\begin{equation} \label{eq:inequality}
\|\mathcal S^{\otimes 2}(\rho)\|_{\ell_2^{d^2} \otimes_\pi \ell_2^{d^2}} \geq \frac{\|\mathcal R^{\otimes 2}(\rho)\|_{\ell_2^{d^2} \otimes_\pi \ell_2^{d^2}} + 1}{2}.
\end{equation}
\end{theorem}

Note that inequality \eqref{eq:inequality} was proven to be an equality for several classes of states in Sections \ref{sec:pure-states}, \ref{sec:symmetric-states} and \ref{sec:noisy-states}.
We show next that it is not the case in general. To this end, consider a product state $\rho = \rho_1 \otimes \rho_2$, for quantum states $\rho_{1,2}$ having respective purities $p_{1,2} := \Tr(\rho_{1,2}^2)$. We then have
\begin{align*}
    \|\mathcal S^{\otimes 2}(\rho)\|_{\ell_2^{d^2} \otimes_\pi \ell_2^{d^2}} &= \sqrt{\Tr\left( \frac{I+F}{2} \rho_1^{\otimes 2} \right) \Tr\left( \frac{I+F}{2} \rho_2^{\otimes 2} \right)} = \frac {\sqrt{(1+p_1)(1+p_2)}}{2},\\
    \|\mathcal R^{\otimes 2}(\rho)\|_{\ell_2^{d^2} \otimes_\pi \ell_2^{d^2}} &= \sqrt{\Tr\left( F \rho_1^{\otimes 2} \right) \Tr\left( F \rho_2^{\otimes 2} \right)} =  \sqrt{p_1p_2}.
\end{align*}
Hence, \eqref{eq:inequality} is saturated if and only if $p_1 = p_2$.

Before proving Theorem \ref{thm:proof-conjecture}, we show a key lemma. 

\begin{lemma}\label{lem:deformed-unitary}
	Let $\{a_1, \ldots, a_n\}$, $\{b_1, \ldots, b_n\}$ be two orthonormal bases of $\mathbb C^n$. For complex numbers $\gamma_1, \ldots, \gamma_n$ such that $|\gamma_i|\geq 1$ for all $1 \leq i \leq n$, define the matrix 
	$$S := \sum_{i=1}^n \gamma_i \ket{a_i}\bra{b_i}.$$
	Then, for any $X \in \mathcal M_n(\mathbb C)$, we have
	\begin{equation}\label{eq:ineq-X}
	\|SXS^*\|_1 \geq \|X\|_1 + \sum_{i=1}^n (|\gamma_i|^2-1)\bra{b_i} X \ket{b_i}.
	\end{equation}
\end{lemma}
\begin{proof}
	First, note that the matrix $S$ is invertible, with inverse
	$$S^{-1} = \sum_{i=1}^n \gamma_i^{-1} \ket{b_i}\bra{a_i}.$$
	Writing $Y := SXS^*$, equation \eqref{eq:ineq-X} is equivalent to
	\begin{equation}\label{eq:ineq-Y}
	\|Y\|_1 \geq \|S^{-1}Y(S^*)^{-1}\|_1 + \sum_{i=1}^n \left( 1 - |\gamma_i|^{-2}\right)\bra{a_i} Y \ket{a_i}.
	\end{equation}
	Note that the right hand side of the inequality above is equal to $\|\Phi(Y)\|_1$, where the map $\Phi : \mathcal M_n(\mathbb C) \to \mathcal M_{2n}(\mathbb C)$ is given by
	$$\Phi(Y) = \left(S^{-1}Y(S^*)^{-1}\right) \oplus \left( \bigoplus_{i=1}^n \left( 1 - |\gamma_i|^{-2}\right)\bra{a_i} Y \ket{a_i} \right).$$
	Hence, equation \eqref{eq:ineq-Y} reads $\|\Phi(Y)\|_1 \leq \|Y\|_1$, which is true if $\Phi$ is a quantum channel (this is a simple consequence of the Russo-Dye theorem, as explained in \cite{perez2006contractivity}). Let us prove next that $\Phi$ is indeed a quantum channel. We have $\Phi(X) = KXK^* + \sum_{i=1}^n L_i X L_i^*$, where 
	$$K = \begin{bmatrix}
	S^{-1} \\ 0_n
	\end{bmatrix}
	\quad \text{ and } \quad 
	L_i = \begin{bmatrix}
	0_n \\ \sqrt{1 - |\gamma_i|^{-2}}\ket{i} \bra{a_i}
	\end{bmatrix}, \ 1\leq i\leq n.$$
	So the fact that $\Phi$ is completely positive is clear.
	And the trace preserving condition is also easily shown to be true, since
	$$K^*K + \sum_{i=1}^n L_i^*L_i = (S^{-1})^*S^{-1} + \sum_{i=1}^n \left(1 - |\gamma_i|^{-2}\right)\ket{a_i}\bra{a_i} =  \sum_{i=1}^n \ket{a_i}\bra{a_i} = I_n.$$
\end{proof}

We can now give the proof of the main result of this section.

\begin{proof}[Proof of Theorem \ref{thm:proof-conjecture}]
	Setting $X := \mathcal R^{\otimes 2}(\rho)$, where $X$ is viewed as belonging to $\mathcal M_{d^2}(\mathbb C)$, it is easy to see that we have
	\begin{align*}
	\mathcal S^{\otimes 2}(\rho) &= \frac{d+1}{2d} \sum_{k,l=1}^{d^2} \bra{x_k \otimes x_l} \rho \ket{x_k \otimes x_l} \ket{kl}\\
	&= \frac{d+1}{2d} \sum_{k,l=1}^{d^2} \bra{x_k \otimes \bar x_k} X \ket{\bar x_l \otimes x_l} \ket{kl}\\
	&= \frac{d+1}{2d} \sum_{k,l=1}^{d^2} \bra{x_k \otimes \bar x_k} X F  \ket{x_l \otimes \bar x_l} \ket{kl}
	\end{align*}
	where $F$ is the flip operator.  We now have $\|\mathcal R^{\otimes 2}(\rho)\|_{\ell_2^{d^2} \otimes_\pi \ell_2^{d^2}} = \|X\|_1$ and $\|\mathcal S^{\otimes 2}(\rho)\|_{\ell_2^{d^2} \otimes_\pi \ell_2^{d^2}} = \|\hat S(XF)\hat S^*\|_1$ \footnote{Note that this norm relation holds in general for any tester $\mathcal E$: $\|\mathcal E^{\otimes 2}(\rho)\|_{\ell_2^{d^2} \otimes_\pi \ell_2^{d^2}} = \|\hat E X \hat E^*\|_1$.}, where $\hat S$ is the matrix of the operator $\mathcal S$, i.e.
	$$\hat S = \sqrt\frac{d+1}{2d} \sum_{k=1}^{d^2} \ket{k}\bra{x_k \otimes \bar x_k}.$$
	Noticing that $1 = \operatorname{Tr} \rho = d \bra{\psi} X \ket{\psi}$, where $\ket{\psi}:=\sum_{i=1}^{d}\ket{ii}/\sqrt{d}$ is the maximally entangled state on $\mathbb C^d \otimes \mathbb C^d$, the inequality in the statement reads 
	$$\|\hat S (XF) {\hat S}^*\|_1 \geq \frac{\|X\|_1 + d \bra{\psi} X \ket{\psi}}{2} = \frac{\|XF\|_1 + d \bra{\psi} XF \ket{\psi}}{2}.$$
	In order to conclude, we need to show that $S := \sqrt 2 \hat S$ can be written as in equation \eqref{eq:ineq-X} from Lemma \ref{lem:deformed-unitary}, with $b_1 = \psi$, $\gamma_1 = \sqrt{d+1}$, and $\gamma_2 = \cdots = \gamma_{d^2} = 1$. 
	
	Indeed, let us set 
	$$\ket{y_k} := \sqrt\frac{d+1}{d} \ket{x_k \otimes \bar x_k} - \frac{\sqrt{d+1}-1}{d}\ket{\psi}.$$ 
	One can show, using the fact that the $\ket{x_k}\bra{x_k}$'s form a SIC POVM, that the $\ket{y_k}$'s form an orthonormal basis of $\mathbb C^{d^2}$ (for similar ideas, see \cite{gour2014construction,jivulescu2017symmetric}). We have thus
	$$S = \sqrt\frac{d+1}{d}\sum_{k=1}^{d^2} \ket{k} \bra{x_k \otimes \bar x_k} = (\sqrt{d+1}-1)\ket{v}\bra{\psi} + \underbrace{\sum_{k=1}^{d^2} \ket{k} \bra{y_k}}_{=:U},$$
	where $\ket v$ is the normalized all-ones vector from \eqref{eq:def-v}, and $U$ is a unitary operator. One can see by direct computation that $U \ket \psi= \ket v$, so we can write 
	$$S = \sqrt{d+1} \ket v \bra \psi + V,$$
	where $V$ is a partial isometry mapping $(\mathbb C \ket \psi)^\perp$ to $(\mathbb C \ket v)^\perp$. We have thus shown that $S$ can be written as in \eqref{eq:ineq-X}, finishing the proof.
\end{proof}

\section{Completeness of the family of criteria in the bipartite case}
\label{sec:completeness}

Our goal in this section is to prove that, in the bipartite case, the family of entanglement criteria that we are looking at is \emph{complete}. What we mean by this is that, given an entangled bipartite state, there always exist testers detecting its entanglement. We have already seen that this is the case for bipartite pure states, and we shall prove a similar result for multipartite pure states in Section \ref{sec:multipartite}. To be fully rigorous, what we are able to show in the case of bipartite mixed states is that our family of entanglement criteria is complete at least when extended to allow for a permutation of the indices before applying the testers, as described in Section \ref{sec:interest}. More precisely, we will prove the following result.

\begin{theorem} \label{th:completeness}
Let $\rho$ be an entangled state on $\mathbb C^d\otimes\mathbb C^d$. Then, there exists a tester $\mathcal{E}:S_1^d\to\ell_2^{d^2}$ such that
\[ \left\| \mathcal E^{\sharp} \otimes \mathcal E\left(F\rho^{\Gamma}\right) \right\|_{\ell_2^{d^2}\otimes_{\pi}\ell_2^{d^2}} > 1 , \]
where $\mathcal{E}^{\sharp}:S_1^d\to\ell_2^{d^2}$ is the tester whose operators are the adjoints of those of $\mathcal{E}$. 
\end{theorem}

Concretely, $F\rho^{\Gamma}$ is the following permutation of indices of $\rho$:
\[ \rho=\sum_{i,j,k,l=1}^d\rho_{ij,kl}\kb{ij}{kl} \implies F\rho^{\Gamma}=\sum_{i,j,k,l=1}^d\rho_{ij,kl}\kb{li}{kj} . \]

Before we launch into the proof of Theorem \ref{th:completeness}, let us make two basic but useful observations. 

By duality we know that if $\rho$ is entangled, i.e.~$\|\rho\|_{S_1^{d} \otimes_\pi S_1^{d}}>1$, then there exists $\Theta$ such that $\|\Theta\|_{S_{\infty}^{d} \otimes_\epsilon S_{\infty}^{d}}\leq 1$ and $\mathrm{Tr}(\Theta\rho)>1$. Now, we can assume without loss of generality that $\Theta$ is Hermitian. This is because $\hat{\Theta}:=(\Theta+\Theta^*)/2$, which is Hermitian, is also such that $\|\hat{\Theta}\|_{S_{\infty}^{d} \otimes_\epsilon S_{\infty}^{d}}\leq 1$ and $\mathrm{Tr}(\hat{\Theta}\rho)>1$. Indeed on the one hand, 
\begin{align*}  
\mathrm{Tr}(\hat{\Theta}\rho) & = \frac{1}{2}\left( \mathrm{Tr}(\Theta\rho) + \mathrm{Tr}(\Theta^*\rho) \right) \\
& = \frac{1}{2}\left( \mathrm{Tr}(\Theta\rho) + \overline{\mathrm{Tr}(\Theta\rho^*)} \right) \\
& = \frac{1}{2}\left( \mathrm{Tr}(\Theta\rho) + \overline{\mathrm{Tr}(\Theta\rho)} \right) \\
& = \mathrm{Tr}(\Theta\rho) \\
& > 1 .
\end{align*}
And on the other hand, for all $X,Y$ such that $\|X\|_1,\|Y\|_1\leq 1$,
\begin{align*} 
\left| \mathrm{Tr}(\hat{\Theta}X\otimes Y) \right| & = \left| \frac{1}{2}\left( \mathrm{Tr}(\Theta X\otimes Y) + \mathrm{Tr}(\Theta^* X\otimes Y) \right) \right| \\
& = \left| \frac{1}{2}\left( \mathrm{Tr}(\Theta X\otimes Y) + \overline{ \mathrm{Tr}(\Theta X^*\otimes Y^*) } \right) \right| \\
& \leq \frac{1}{2}\left( \left| \mathrm{Tr}(\Theta X\otimes Y) \right| + \left| \mathrm{Tr}(\Theta X^*\otimes Y^*) \right| \right) \\
& \leq \|\Theta\|_{S_{\infty}^{d} \otimes_\epsilon S_{\infty}^{d}} \\
& \leq 1 .
\end{align*}

What is more, we can also assume without loss of generality that $\Theta$ is positive semidefinite. This is because $\Theta_{\lambda}:=\lambda I+(1-\lambda)\Theta$, which is positive semidefinite for $0\leq\lambda\leq 1$ large enough, is also such that $\|\Theta_{\lambda}\|_{S_{\infty}^{d} \otimes_\epsilon S_{\infty}^{d}}\leq 1$ and $\mathrm{Tr}(\Theta_{\lambda}\rho)>1$. Indeed on the one hand, 
\[ \mathrm{Tr}(\Theta_{\lambda}\rho) = \lambda\mathrm{Tr}(\rho) + (1-\lambda)\mathrm{Tr}(\Theta\rho) > \lambda + (1-\lambda) = 1 . \]
And on the other hand, for all $X,Y$ such that $\|X\|_1,\|Y\|_1\leq 1$,
\begin{align*} 
\left| \mathrm{Tr}(\Theta_{\lambda} X\otimes Y) \right| & = \left| \lambda \mathrm{Tr}(X\otimes Y) + (1-\lambda)\mathrm{Tr}(\Theta X\otimes Y) \right| \\
& \leq \lambda \left| \mathrm{Tr}(X)\mathrm{Tr}(Y) \right| + (1-\lambda) \left| \mathrm{Tr}(\Theta X\otimes Y) \right| \\
& \leq \lambda \|X\|_1\|Y\|_1 + (1-\lambda) \|\Theta\|_{S_{\infty}^{d} \otimes_\epsilon S_{\infty}^{d}}  \\
& \leq \lambda + (1-\lambda) \\
& = 1 .
\end{align*}

We are now ready to prove Theorem \ref{th:completeness}.

\begin{proof}
Let $\Theta$ be such that $\|\Theta\|_{S_{\infty}^{d} \otimes_\epsilon S_{\infty}^{d}}\leq 1$ and $\mathrm{Tr}(\Theta\rho)>1$. As justified above, we can assume without loss of generality that $\Theta^*=\Theta$ and $\Theta\geq 0$. From Lemma \ref{lem:T-Theta}, we thus know that $T:=\Theta^{\Gamma}F$ is a test operator. This means that there exist operators $\{E_k\}_{k=1}^{d^2}$ on $\mathbb C^d$ such that $T=\sum_{k=1}^{d^2}E_k\otimes E_k^*$ and $\mathcal E:X\in\mathcal M_d(\mathbb C^d)\mapsto\sum_{k=1}^{d^2}\mathrm{Tr}(E_k^*X)\ket{k}\in\mathbb C^{d^2}$ is a tester.

Let us now prove that $\| \mathcal E^{\sharp} \otimes \mathcal E(F\rho^{\Gamma}) \|_{\ell_2^{d^2}\otimes_{\pi}\ell_2^{d^2}} > 1$. We have
\[ \mathcal E^{\sharp} \otimes \mathcal E\left(F\rho^{\Gamma}\right) = \sum_{k,l=1}^{d^2} \mathrm{Tr}\left(E_k\otimes E_l^*F\rho^{\Gamma}\right) \ket{kl} . \]
Next, observe that $\ket{u}:=\sum_{k=1}^{d^2}\ket{kk}$ is such that $\|u\|_{\ell_2^{d^2}\otimes_{\epsilon}\ell_2^{d^2}}=1$. Hence by duality,
\begin{align*}
\left\| \mathcal E^{\sharp} \otimes \mathcal E\left(F\rho^{\Gamma}\right) \right\|_{\ell_2^{d^2}\otimes_{\pi}\ell_2^{d^2}} & \geq \bk{u}{\mathcal E^* \otimes \mathcal E(F\rho^{\Gamma})} \\
& = \sum_{k=1}^{d^2} \mathrm{Tr}\left(E_k\otimes E_k^*F\rho^{\Gamma}\right) \\
& = \mathrm{Tr}\left(TF\rho^{\Gamma}\right) \\
& = \mathrm{Tr}\left((TF)^{\Gamma}\rho\right) \\
& = \mathrm{Tr}(\Theta\rho) \\
& > 1 ,
\end{align*}
which is exactly what we wanted to show.
\end{proof}

\section{Entanglement testers in the multipartite setting}
\label{sec:multipartite}

In this section we discuss the power of the realignment entanglement tester, when used on multipartite pure quantum states. In the first subsection, we show that the criterion obtained by applying several copies of the realignment tester detects \emph{all} pure entangled multipartite states. In the following two subsections we show that the multipartite realignment criterion is, in a sense, optimal in the case of the so-called $W$ state and in the case of pure multipartite states admitting a generalized Schmidt decomposition.

\subsection{Entanglement detection of multipartite pure states by the realignment tester} \label{sec:multi-general} \hfill\smallskip

We consider the entanglement criterion on $(\mathbb{C}^d)^{\otimes m}$ defined by the realignment map $\mathcal{R}^{\otimes m}$, namely: for any state $\rho$ on $(\mathbb{C}^d)^{\otimes m}$,
\[ \left\| \mathcal{R}^{\otimes m}(\rho) \right\|_{(\ell_2^{d^2})^{\otimes_{\pi}m}} > 1 \ \Longrightarrow\ \rho \ \text{entangled} . \]
We want to show that the above implication is an equivalence on the set of pure states. More precisely, we will establish the following result.

\begin{theorem} \label{th:multi-nsc}
For any unit vector $\varphi\in (\mathbb{C}^d)^{\otimes m}$,
\begin{equation} \label{eq:lb-pi-final} 
\left\| \hat{\varphi} \right\|_{(\ell_2^{d^2})^{\otimes_{\pi}m}} \geq \frac{1}{\left\| \varphi \right\|_{(\ell_2^{d})^{\otimes_{\epsilon}m}}} . 
\end{equation}
If in addition $\varphi$ is non-negative (meaning that its coefficients in the canonical basis of $(\mathbb{C}^d)^{\otimes m}$ are all non-negative), then
\begin{equation} \label{eq:lb-pi-nn} 
\left\| \hat{\varphi} \right\|_{(\ell_2^{d^2})^{\otimes_{\pi}m}} \geq \frac{1}{\left\| \varphi \right\|_{(\ell_2^{d})^{\otimes_{\epsilon}m}}^2} . \end{equation}
\end{theorem}

As an immediate consequence of Theorem \ref{th:multi-nsc}, we have by Proposition \ref{prop:sep-pure}
\[ \varphi\ \text{entangled}\ \Longrightarrow\  \left\| \varphi \right\|_{(\ell_2^d)^{\otimes_{\epsilon}m}} < 1\ \Longrightarrow\ \left\| \hat{\varphi} \right\|_{(\ell_2^{d^2})^{\otimes_{\pi}m}} >1 . \]
So we indeed have shown that
\[ \varphi\ \text{entangled}\ \Longleftrightarrow\  \left\| \mathcal{R}^{\otimes m}(\kb{\varphi}{\varphi}) \right\|_{(\ell_2^{d^2})^{\otimes_{\pi}m}} >1 . \]

\begin{proof}
Define $\hat{\varphi}\in (\mathbb{C}^{d^2})^{\otimes m}$ as
\[ \ket{\hat{\varphi}} := \mathcal{R}^{\otimes m}(\kb{\varphi}{\varphi}) = \sum_{\substack{1\leq i_1,\ldots,i_m \leq d \\ 1\leq j_1,\ldots,j_m \leq d}} \bk{j_1\cdots j_m}{\varphi} \bk{\varphi}{i_1\cdots i_m} \ket{i_1j_1\cdots i_mj_m} . \]
By duality, we know that
\begin{equation} \label{eq:duality} 
\left\| \hat{\varphi} \right\|_{(\ell_2^{d^2})^{\otimes_{\pi}m}} \geq \frac{\bk{\hat{\varphi}}{\hat{\varphi}}}{\left\| \hat{\varphi} \right\|_{(\ell_2^{d^2})^{\otimes_{\epsilon}m}}} . 
\end{equation}
First observe that, for any unit vectors $a^1,\ldots,a^m\in\mathbb{C}^{d^2}$,
\[ \bk{\hat{\varphi}}{a^1\cdots a^m} = \sum_{\substack{1\leq i_1,\ldots,i_m \leq d \\ 1\leq j_1,\ldots,j_m \leq d}} \varphi_{j_1\ldots j_m} \bar{\varphi}_{i_1\ldots i_m} a^1_{i_1j_1}\cdots a^m_{i_mj_m} = \bk{\varphi\bar{\varphi}}{a^1\cdots a^m} . \]
Therefore, 
\begin{align*}
\left\| \hat{\varphi} \right\|_{(\ell_2^{d^2})^{\otimes_{\epsilon}m}} & = \max \left\{ \bk{\hat{\varphi}}{a^1\cdots a^m},\ a^1,\ldots,a^m\in\mathbb{C}^{d^2},\, \|a^1\|,\ldots,\|a^m\| \leq 1 \right\} \\
& = \max \left\{ \bk{\varphi\bar{\varphi}}{a^1\cdots a^m},\ a^1,\ldots,a^m\in\mathbb{C}^{d^2},\, \|a^1\|,\ldots,\|a^m\| \leq 1 \right\} \\
& = \left\| \varphi\otimes \bar{\varphi} \right\|_{(\ell_2^{d^2})^{\otimes_{\epsilon}m}} .
\end{align*}
Second we have
\[ \bk{\hat{\varphi}}{\hat{\varphi}} = \sum_{\substack{1\leq i_1,\ldots,i_m \leq d \\ 1\leq j_1,\ldots,j_m \leq d}} \varphi_{j_1\ldots j_m} \bar{\varphi}_{i_1\ldots i_m} \bar{\varphi}_{j_1\ldots j_m} \varphi_{i_1\ldots i_m} =  |\bk{\varphi}{\varphi} |^2 = 1 . \]
Hence, inserting the two above equalities into equation \eqref{eq:duality}, we get
\begin{equation} \label{eq:lb-pi} 
\left\| \hat{\varphi} \right\|_{(\ell_2^{d^2})^{\otimes_{\pi}m}} \geq \frac{1}{\left\| \varphi \otimes \bar{\varphi} \right\|_{(\ell_2^{d^2})^{\otimes_{\epsilon}m}}} . 
\end{equation}

In the case where $\varphi$ is non-negative, then first $\bar{\varphi}=\varphi$. And second we know from \cite[Theorem 5]{zhu2010additivity} that its geometric measure of entanglement (as defined in equation \eqref{eq:G}) is additive, which is equivalent to saying that its injective $\ell_2$-norm is multiplicative. This means that
$$ \left\| \varphi \otimes \bar{\varphi} \right\|_{(\ell_2^{d^2})^{\otimes_{\epsilon}m}} = \left\| \varphi \otimes \varphi \right\|_{(\ell_2^{d^2})^{\otimes_{\epsilon}m}} = \left\| \varphi \right\|_{(\ell_2^{d})^{\otimes_{\epsilon}m}} ^2 .$$
Inequality \eqref{eq:lb-pi-nn} is thus proven.

To deal with the general case, let us define, for any state $\rho$ on $H_1\otimes\cdots\otimes H_m$,
\begin{align*} 
h_{sep(H_1:\cdots:H_m)}(\rho) :=\, & \max\left\{ \mathrm{Tr}(\rho\, \sigma^1\otimes\cdots\otimes\sigma^m),\, \sigma^k \text{ state on } H_k,\, 1\leq k\leq m \right\} \\
=\, & \max\left\{ \mathrm{Tr}(\rho\, \kb{a^1}{a^1} \otimes\cdots\otimes \kb{a^m}{a^m}),\, a^k\in H_k,\, \|a^k\|=1, \, 1\leq k\leq m \right\} ,
\end{align*}
where the last equality is by extremality of pure product states amongst product states. We thus see that, for any unit vector $\chi\in H_1\otimes\cdots\otimes H_m$,
\[ \|\chi\|_{H_1\otimes_{\epsilon}\cdots\otimes_{\epsilon}H_m} = \sqrt{h_{sep(H_1:\cdots:H_m)}(\kb{\chi}{\chi})} . \]
Now, let $\rho,\rho'$ be states on $H_1\otimes\cdots\otimes H_m$, and assume that $\sigma^1,\ldots,\sigma^m$ are states on $H_1\otimes H_1,\ldots,H_m\otimes H_m$ such that
\[ h_{sep(H_1H_1:\cdots:H_mH_m)}(\rho\otimes\rho') = \mathrm{Tr}(\rho\otimes\rho' \sigma^1\otimes\cdots\otimes\sigma^m) . \]
We then have, denoting by $\tilde{\sigma}^k:=\mathrm{Id}\otimes\mathrm{Tr}(\sigma^k)$ the reduced state of $\sigma^k$ on $H_k$,
\begin{align*}
h_{sep(H_1:\cdots:H_m)}(\rho) & \geq \mathrm{Tr}(\rho\, \tilde{\sigma}^1\otimes\cdots\otimes\tilde{\sigma}^m) \\
& = \mathrm{Tr}(\rho\otimes I\, \sigma^1\otimes\cdots\otimes\sigma^m) \\
& \geq \mathrm{Tr}(\rho\otimes\rho' \sigma^1\otimes\cdots\otimes\sigma^m) \\
& = h_{sep(H_1H_1:\cdots:H_mH_m)}(\rho\otimes\rho') .
\end{align*}
This implies that, for any unit vectors $\chi,\chi'\in H_1\otimes\cdots\otimes H_m$,
\begin{equation} \label{eq:chi-chi'} \|\chi\|_{H_1\otimes_{\epsilon}\cdots\otimes_{\epsilon}H_m} \geq \|\chi\otimes \chi'\|_{H_1H_1\otimes_{\epsilon}\cdots\otimes_{\epsilon}H_mH_m} . \end{equation}
Coming back to equation \eqref{eq:lb-pi}, we eventually get using the above observation that
$$ \left\| \hat{\varphi} \right\|_{(\ell_2^{d^2})^{\otimes_{\pi}m}} \geq \frac{1}{\left\| \varphi \right\|_{(\ell_2^{d})^{\otimes_{\epsilon}m}}} , $$
which is exactly inequality \eqref{eq:lb-pi-final}.
\end{proof}

\begin{remark}
Note that, in the case of a not necessarily normalized vector $\varphi\in(\mathbb C^d)^{\otimes m}$, equation \eqref{eq:lb-pi} would actually take the form
$$ \left\| \hat{\varphi} \right\|_{(\ell_2^{d^2})^{\otimes_{\pi}m}} \geq \frac{\|\varphi\|_2^4}{\left\| \varphi \otimes \bar{\varphi} \right\|_{(\ell_2^{d^2})^{\otimes_{\epsilon}m}}} . $$
Now, in the case where the vectors $\chi,\chi'\in H_1\otimes\cdots\otimes H_m$ are not necessarily normalized, equation \eqref{eq:chi-chi'} would read instead
$$ \|\chi\|_{H_1\otimes_{\epsilon}\cdots\otimes_{\epsilon}H_m} \geq \frac{ \|\chi\otimes \chi'\|_{H_1H_1\otimes_{\epsilon}\cdots\otimes_{\epsilon}H_mH_m} }{ \|\chi'\|_{2} } . $$
This implies that, in the case of a not necessarily normalized vector $\varphi\in(\mathbb C^d)^{\otimes m}$, equation \eqref{eq:lb-pi-final} would be 
$$ \left\| \hat{\varphi} \right\|_{(\ell_2^{d^2})^{\otimes_{\pi}m}} \geq \frac{\|\varphi\|_2^4}{\|\varphi\|_2 \left\| \varphi \right\|_{(\ell_2^{d})^{\otimes_{\epsilon}m}}} = \frac{\|\varphi\|_2^3}{\left\| \varphi \right\|_{(\ell_2^{d})^{\otimes_{\epsilon}m}}} . $$
\end{remark}

It is instructive to see what the lower bound \eqref{eq:lb-pi-final} gives in the bipartite case. Indeed, in this case we can compute the exact values of the quantities on the left and right hand sides of the inequality. Namely, if a unit vector $\varphi\in(\mathbb{C}^d)^{\otimes 2}$ has Schmidt decomposition
\[ \ket{\varphi}=\sum_{i=1}^r\sqrt{\lambda_i}\ket{e_if_i} , \] 
then we have on the one hand
\[ \left\|  \mathcal{R}^{\otimes 2}(\kb{\varphi}{\varphi}) \right\|_{(\ell_2^{d^2})^{\otimes_{\pi}2}} = \left(\sum_{i=1}^r\sqrt{\lambda_i}\right)^2 , \]
and on the other hand
\[ \left\| \varphi \right\|_{(\ell_2^{d})^{\otimes_{\epsilon}2}} = \sqrt{\lambda_1} . \]
So the largest gap in inequality \eqref{eq:lb-pi-final} is when $\varphi$ has uniform Schmidt coefficients, in which case the left hand side is equal to $r$ while the right hand side is equal to $\sqrt{r}$.

\subsection{The example of the \texorpdfstring{$W$}{W} state} \hfill\smallskip

As an illustration of the power of the entanglement criterion based on the realignment tester on multipartite pure states, let us see what it yields when applied to the famous $W$ state, known to be the maximally entangled three-qubit pure state \cite{derksen2017}. We recall that the latter is defined as
$$ \ket{w} := \frac{1}{\sqrt{3}}\left( \ket{112} + \ket{121} + \ket{211} \right) \in (\mathbb C^2)^{\otimes 3} .$$
It is entangled, and we know from \cite[Lemma 6.2]{Friedland2018} that
$$ \| \kb{w}{w} \|_{(S_1^2)^{\otimes_{\pi}3}} = \| w \|_{(\ell_2^2)^{\otimes_{\pi}3}}^2 = \left(\frac{3}{2}\right)^2 = \frac{9}{4} .$$

We would now like to compare the above value to the value of the $(\ell_2^4)^{\otimes_{\pi}3}$ norm of $\mathcal{R}^{\otimes 3}(\kb{w}{w})$. Since $w$ is non-negative, inequality \eqref{eq:lb-pi-nn} tells us that
$$ \left\| \mathcal{R}^{\otimes 3}(\kb{w}{w}) \right\|_{(\ell_2^4)^{\otimes_{\pi}3}} \geq \frac{1}{\|w\|_{(\ell_2^2)^{\otimes_{\epsilon}3}}^2} .$$
Now, we know from \cite[Lemma 6.2]{Friedland2018} again that
$$ \| w \|_{(\ell_2^2)^{\otimes_{\epsilon}3}} = \frac{2}{3} . $$
We thus have
$$ \left\| \mathcal{R}^{\otimes 3}(\kb{w}{w}) \right\|_{(\ell_2^4)^{\otimes_{\pi}3}} \geq \frac{1}{\left(2/3\right)^2} = \frac{9}{4}, $$
which is actually an equality since on the other hand
$$ \left\| \mathcal{R}^{\otimes 3}(\kb{w}{w}) \right\|_{(\ell_2^4)^{\otimes_{\pi}3}} \leq \| \kb{w}{w} \|_{(S_1^2)^{\otimes_{\pi}3}} = \frac{9}{4} .$$

To summarize, we have shown that the realignment tester $\mathcal R$ optimally detects the entanglement of the $W$ state, in the sense that
$$ \left\| \mathcal{R}^{\otimes 3}(\kb{w}{w}) \right\|_{(\ell_2^4)^{\otimes_{\pi}3}} = \| \kb{w}{w} \|_{(S_1^2)^{\otimes_{\pi}3}} = \frac{9}{4} > 1 .$$
What is more, this is an example where our lower bound \eqref{eq:lb-pi-nn} is tight.

\begin{remark}
In general, we know from Theorem \ref{th:multi-nsc} that, for any unit vector $\phi\in(\mathbb C^d)^{\otimes m}$, the following inequalities hold: $$\|\phi\|^2_{(\ell_2^d)^{\otimes_\pi m}} = \| \kb{\phi}{\phi}\|_{(S_1^d)^{\otimes_\pi m}} \geq \|\mathcal R^{\otimes m}(\kb{\phi}{\phi})\|_{(\ell_2^{d^2})^{\otimes_\pi m}}  \geq \frac{1}{\|\phi\|_{(\ell_2^d)^{\otimes_\epsilon m}}}.$$
This shows that for all unit vectors $\phi\in(\mathbb C^d)^{\otimes m}$ for which $\|\phi\|_\pi^2 \|\phi\|_\epsilon = 1$, the realignment criterion is exact. Similarly we also have that the realignment criterion is exact for all non-negative unit vectors $\varphi$ such that $\|\phi\|_\pi \|\phi\|_\epsilon = 1$.

Note that the unit vector
$$\ket{v} = \frac 1 2 (\ket{112} + \ket{121} + \ket{211} - \ket{222}) \in (\mathbb C^2)^{\otimes 3},$$
studied in \cite[Lemma 6.1]{Friedland2018}, saturates the duality relation $\|v\|_\pi \|v\|_\epsilon = 1$. But it has negative coefficients, so we cannot guarantee the exactness of the realignment criterion in this case. 
\end{remark}

\subsection{The case of multipartite pure states having a generalized Schmidt decomposition} \hfill\smallskip

We now focus on the particular case where the unit vector $\varphi\in(\mathbb C^d)^{\otimes m}$ admits a generalized Schmidt decomposition, i.e.
\[ \ket{\varphi} = \sum_{k=1}^r \sqrt{\lambda_k} \ket{e^1_k\cdots e^m_k} , \]
where $\lambda_1,\ldots,\lambda_r>0$ are such that $\sum_{k=1}^r\lambda_k=1$ and $\{e^1_k\}_{k=1}^r,\ldots,\{e^m_k\}_{k=1}^r$ are orthonormal families in $\mathbb C^d$. 

For such multipartite pure state $\varphi$, setting again $\ket{\hat{\varphi}}:=\mathcal R^{\otimes m}(\kb{\varphi}{\varphi})$, we have,
\begin{align*}
\ket{\hat{\varphi}} & = \sum_{k,l=1}^r \sqrt{\lambda_k\lambda_l} \sum_{\substack{1\leq i_1,\ldots,i_m \leq d \\ 1\leq j_1,\ldots,j_m \leq d}} \bk{j_1\cdots j_m}{e^1_k\cdots e^m_k} \bk{e^1_l\cdots e^m_l}{i_1\cdots i_m} \ket{i_1j_1\cdots i_mj_m} \\
& = \sum_{k,l=1}^r \sqrt{\lambda_k\lambda_l} \sum_{i_1=1}^d \bk{e^1_l}{i_1}\ket{i_1} \sum_{j_1=1}^d \bk{j_1}{e^1_k}\ket{j_1} \cdots \sum_{i_m=1}^d \bk{e^m_l}{i_m}\ket{i_m} \sum_{j_m=1}^d \bk{j_m}{e^m_k}\ket{j_m} . 
\end{align*}  
Now, we just have to observe that, for each $1\leq k,l\leq r$ and $1\leq q\leq m$,
\[ \sum_{i_q=1}^d \bk{e^q_l}{i_q}\ket{i_q}=\ket{\bar{e}^q_l} \text{ and } \sum_{j_q=1}^d \bk{j_q}{e^q_k}\ket{j_q}=\ket{e^q_k} .\]
Therefore, we actually have
\[ \ket{\hat{\varphi}} = \sum_{k,l=1}^r \sqrt{\lambda_k\lambda_l}\ket{\bar{e}^1_k e^1_l\cdots \bar{e}^m_k e^m_l} . \]
We recognize in the expression above a generalized Schmidt decomposition in $(\mathbb C^{d^2})^{\otimes m}$. Hence,
\[ \left\|\hat{\varphi}\right\|_{(\ell_2^{d^2})^{\otimes_{\pi}m}} = \sum_{k,l=1}^r \sqrt{\lambda_k\lambda_l} . \]

Note that, for such multipartite state $\varphi$, we know from \cite[Proof of Theorem 0.1]{palazuelos2014largest} that
\[ \| \kb{\varphi}{\varphi} \|_{(S_1^d)^{\otimes_{\pi} m}} = \sum_{k,l=1}^r \sqrt{\lambda_k\lambda_l} . \]
So, just as in the bipartite case, this is a situation where \[ \left\| \mathcal R^{\otimes m}(\kb{\varphi}{\varphi}) \right\|_{(\ell_2^{d^2})^{\otimes_{\pi}m}} = \| \kb{\varphi}{\varphi} \|_{(S_1^d)^{\otimes_{\pi} m}} . \]

\section{Conclusions and open problems}\label{sec:conclusions}

We have introduced in this work a new paradigm for entanglement detection in bipartite and multipartite quantum systems, based on entanglement testers. The main idea is to reduce the characterization of entanglement, based on computing the projective norm in a tensor product of Schatten $1$-norm classes to that of the projective norm in a tensor product of Hilbert spaces. In other words, using entanglement testers, one reduces the entanglement problem of mixed quantum states to that of pure quantum states (which is know to be much simpler), at the cost of obtaining only a sufficient criterion for entanglement. The most symmetric entanglement testers correspond to the realignment criterion and to the SIC POVM criterion, which have been studied extensively in the literature, in the bipartite case. Our work provides a natural generalization of these criteria to the multipartite setting. 

We analyze the performance of entanglement testers, identifying the important subclass of perfect testers (the ones which detect every pure entangled state). We compare the realignment tester with the SIC POVM tester, showing an exact relation between them in the bipartite case, which allows us to prove two recent conjectures. We then show that every entangled bipartite mixed quantum state can be detected, after index permutation, by a pair of specifically tailored testers, and that the realignment tester can detect any multipartite pure entangled state. 

The main question that remains unanswered at this point is whether our family of entanglement criteria is complete, without allowing for permutation of indices before applying the testers. Concretely, we would like to know if, for any entangled state $\rho$ on $\mathbb C^{d_1}\otimes\cdots\otimes\mathbb C^{d_m}$, there exist testers $\mathcal E_i:S_1^{d_i}\to\ell_2^{n_i}$, $1\leq i\leq m$, such that $\|\mathcal E_1\otimes\cdots\otimes\mathcal E_m(\rho)\|_{\ell_2^{n_1}\otimes_\pi\cdots\otimes_\pi\ell_2^{n_m}}>1$. Even for bipartite states we are only able to show a weaker version of this statement (see Section \ref{sec:completeness}). Note that, in the bipartite case, this problem can be seen as a \emph{factorization through $\ell_2$} problem. Indeed, given an entangled state $\rho$ on $\mathbb C^d\otimes\mathbb C^d$, there exists by definition an operator $\mathcal T:S_1^d\to S_\infty^d$ with norm at most $1$ witnessing its entanglement. And in order to exhibit testers detecting its entanglement, we would need to find operators $\mathcal E,\mathcal F:S_1^d\to\ell_2^n$ with norms at most $1$ such that $\mathcal T=\mathcal E^*\mathcal F$. It can be shown that not every $\mathcal T:S_1^d\to S_\infty^d$ can be factorized in this way with constant $1$ \cite{GuillaumeCIRM}. However, this does not tell us if there exist entangled states that do not have any factorizable entanglement witness.

In a different direction, it would be worth investigating further the performance of our entanglement criteria in the multipartite setting. Indeed, the only quantitative results that we establish in this work when more then two parties are involved are for pure states. But what about the case of mixed states? Are there interesting classes of multipartite mixed states whose entanglement can be detected by the realignment or SIC POVM testers? And can we, in general, compare the respective performances of these two testers?

Finally, it could be interesting to probe the efficiency of entanglement testers $\mathcal E : S_1^{d} \to \mathbb C^{n}$ in the case where the output dimension $n$ is (much) smaller than the dimension of the input space $d^2$. Although such testers cannot be perfect, computing the projective norm of the output tensor is easier when the dimension $n$ is smaller, so the trade-off between the computational efficiency and the performance of these testers needs to be assessed. 

\bigskip

\noindent\textit{Acknowledgements.} MAJ acknowledges the support of Universit\'e Toulouse III Paul Sabatier in the form of an invited professorship, during which this work was initiated. We would like to thank Guillaume Aubrun for the extremely valuable help in understanding map factorization questions.

\bibliographystyle{alpha}
\bibliography{tensors}

\newcommand{\etalchar}[1]{$^{#1}$}
\begin{thebibliography}{PGWPR06}

\bibitem[AFZ13]{appleby}
David~M Appleby, Christopher Fuchs, and Huangjun Zhu.
\newblock Group theoretic, {L}ie algebraic and {J}ordan algebraic formulations
  of the {SIC} existence problem.
\newblock {\em Quantum Information and Computation}, 15, 12 2013.

\bibitem[AS17]{aubrun2017alice}
Guillaume Aubrun and Stanis{\l}aw~J Szarek.
\newblock {\em Alice and Bob Meet Banach: The Interface of Asymptotic Geometric
  Analysis and Quantum Information Theory}, volume 223.
\newblock American Mathematical Society, 2017.

\bibitem[Aub20]{GuillaumeCIRM}
Guillaume Aubrun.
\newblock Personal communication, Oct 2020.

\bibitem[CW03]{chen2003realignment}
Kai Chen and Ling-An Wu.
\newblock A matrix realignment method for recognizing entanglement.
\newblock {\em Quantum Information and Computation}, 3:193--202, 2003.

\bibitem[DFLW17]{derksen2017}
Harm Derksen, Shamuel Friedland, Lek-Hang Lim, and Li~Wang.
\newblock Theoretical and computational aspects of entanglement.
\newblock {\em arXiv preprint:1705.07160v1}, 2017.

\bibitem[FL18]{Friedland2018}
Shmuel Friedland and Lek-Heng Lim.
\newblock Nuclear norm of higher-order tensors.
\newblock {\em Mathematics of Computation}, 87(311):1255--1281, 2018.

\bibitem[GA16]{appleby2016}
Matthew~A Graydon and David~M Appleby.
\newblock Quantum conical designs.
\newblock {\em Journal of Physics A: Mathematical and Theoretical},
  49(8):085301, 2016.

\bibitem[Gha10]{Gharibian}
Sevag Gharibian.
\newblock Strong {NP}-hardness of the quantum separability problem.
\newblock {\em Quantum Information and Computation}, 10(3):343--360, 2010.

\bibitem[GK14]{gour2014construction}
Gilad Gour and Amir Kalev.
\newblock Construction of all general symmetric informationally complete
  measurements.
\newblock {\em Journal of Physics A: Mathematical and Theoretical},
  47(33):335302, 2014.

\bibitem[HHH06]{horodecki2006separability}
Micha{\l} Horodecki, Pawe{\l} Horodecki, and Ryszard Horodecki.
\newblock Separability of mixed quantum states: linear contractions and
  permutation criteria.
\newblock {\em Open Systems \& Information Dynamics}, 13(1):103--111, 2006.

\bibitem[JNG17]{jivulescu2017symmetric}
Maria~Anastasia Jivulescu, Ion Nechita, and Pa{\c{s}}c G{\u{a}}vru{\c{t}}a.
\newblock On symmetric decompositions of positive operators.
\newblock {\em Journal of Physics A: Mathematical and Theoretical},
  50(16):165303, 2017.

\bibitem[Joh11]{johnston2011characterizing}
Nathaniel Johnston.
\newblock Characterizing operations preserving separability measures via linear
  preserver problems.
\newblock {\em Linear and Multilinear Algebra}, 59(10):1171--1187, 2011.

\bibitem[LLFW18]{lai2018entanglement}
Le-Min Lai, Tao Li, Shao-Ming Fei, and Zhi-Xi Wang.
\newblock Entanglement criterion via general symmetric informationally complete
  measurements.
\newblock {\em Quantum Information Processing}, 17(11):314, 2018.

\bibitem[NC10]{nielsen2010quantum}
Michael~A Nielsen and Isaac~L Chuang.
\newblock {\em Quantum computation and quantum information}.
\newblock Cambridge University Press, 2010.

\bibitem[Pal14]{palazuelos2014largest}
Carlos Palazuelos.
\newblock On the largest {B}ell violation attainable by a quantum state.
\newblock {\em Journal of Functional Analysis}, 267(7):1959--1985, 2014.

\bibitem[PG04]{perez2004deciding}
David P{\'e}rez-Garc{\'\i}a.
\newblock Deciding separability with a fixed error.
\newblock {\em Physics Letters A}, 330(3-4):149--154, 2004.

\bibitem[PGM{\etalchar{+}}11]{puchala2011product}
Zbigniew Pucha{\l}a, Piotr Gawron, Jaros{\l}aw~Adam Miszczak, {\L}ukasz
  Skowronek, Man-Duen Choi, and Karol {\.Z}yczkowski.
\newblock Product numerical range in a space with tensor product structure.
\newblock {\em Linear algebra and its applications}, 434(1):327--342, 2011.

\bibitem[PGWPR06]{perez2006contractivity}
David P{\'e}rez-Garc{\'\i}a, Michael~M Wolf, Denes Petz, and Mary~Beth Ruskai.
\newblock Contractivity of positive and trace-preserving maps under ${L}_p$
  norms.
\newblock {\em Journal of Mathematical Physics}, 47(8):083506, 2006.

\bibitem[Rud00]{rudolph2000separability}
Oliver Rudolph.
\newblock A separability criterion for density operators.
\newblock {\em Journal of Physics A: Mathematical and General}, 33(21):3951,
  2000.

\bibitem[Rud04]{rudolph2004ccnr}
Oliver Rudolph.
\newblock Computable cross-norm criterion for separability.
\newblock {\em Letters in Mathematical Physics}, 70:57--64, 2004.

\bibitem[Rud05]{rudolph2005further}
Oliver Rudolph.
\newblock Further results on the cross norm criterion for separability.
\newblock {\em Quantum Information Processing}, 4(3):219--–239, 2005.

\bibitem[SAZG18]{shang2018enhanced}
Jiangwei Shang, Ali Asadian, Huangjun Zhu, and Otfried G{\"u}hne.
\newblock Enhanced entanglement criterion via symmetric informationally
  complete measurements.
\newblock {\em Physical Review A}, 98(2):022309, 2018.

\bibitem[Shi95]{shimony1995degree}
Abner Shimony.
\newblock Degree of entanglement.
\newblock {\em Annals of the New York Academy of Sciences}, 755(1):675--679,
  1995.

\bibitem[SSC20]{sarbicki2020correlation}
Gniewomir Sarbicki, Giovanni Scala, and Dariusz Chru\'{s}ci\'{n}ski.
\newblock Family of multipartite separability criteria based on a correlation
  tensor.
\newblock {\em Phys. Rev. A}, 101:012341, 2020.

\bibitem[TJ89]{tomczak1989banach}
Nicole Tomczak-Jaegermann.
\newblock {\em {B}anach-{M}azur distances and finite-dimensional operator
  ideals}, volume~38.
\newblock Longman Sc \& Tech, 1989.

\bibitem[Wat18]{watrous2018theory}
John Watrous.
\newblock {\em The Theory of Quantum Information}.
\newblock Cambridge University Press, 2018.

\bibitem[WG03]{wei2003geometric}
Tzu-Chieh Wei and Paul~M Goldbart.
\newblock Geometric measure of entanglement and applications to bipartite and
  multipartite quantum states.
\newblock {\em Physical Review A}, 68(4):042307, 2003.

\bibitem[ZCH10]{zhu2010additivity}
Huangjun Zhu, Lin Chen, and Masahito Hayashi.
\newblock Additivity and non-additivity of multipartite entanglement measures.
\newblock {\em New Journal of Physics}, 12(8):083002, 2010.

\end{thebibliography}

\end{document}